\documentclass[fullpage,11pt]{article}

\usepackage{amsthm}
\usepackage{graphicx} 
\usepackage{array} 

\usepackage{amsmath, amssymb, amsfonts, verbatim}
\usepackage{hyphenat,epsfig,subfigure,multirow}

\usepackage[usenames,dvipsnames]{xcolor}
\usepackage[ruled]{algorithm2e}

\usepackage{tcolorbox}

\definecolor{DarkRed}{rgb}{0.5,0.1,0.1}
\definecolor{DarkBlue}{rgb}{0.1,0.1,0.5}

\usepackage{hyperref}
\hypersetup{
colorlinks=true,
pdfnewwindow=true,
citecolor=ForestGreen,
linkcolor=DarkRed,
filecolor=DarkRed,
urlcolor=DarkBlue
}

\usepackage{bm}
\usepackage{url}
\usepackage{xspace} 
\usepackage[mathscr]{euscript}

\usepackage{mdframed}

\usepackage[noend]{algpseudocode}
\makeatletter
\def\BState{\State\hskip-\ALG@thistlm}
\makeatother

\usepackage{cite}
\usepackage{enumerate}

\usepackage[margin=1in]{geometry}

\newtheorem{theorem}{Theorem}
\newtheorem{lemma}{Lemma}[section]
\newtheorem{proposition}[lemma]{Proposition}

\newtheorem{claim}[lemma]{Claim}

\newtheorem{definition}{Definition}
\newtheorem*{Definition}{Definition}
\newtheorem{problem}{Problem}
\newtheorem{remark}[lemma]{Remark}

\newtheorem*{claim*}{Claim}
\newtheorem*{proposition*}{Proposition}
\newtheorem*{lemma*}{Lemma}
\newtheorem*{problem*}{Problem}

\newtheorem{mdresult}{Result}
\newenvironment{result}{\begin{mdframed}[backgroundcolor=lightgray!40,topline=false,rightline=false,leftline=false,bottomline=false,innertopmargin=2pt]\begin{mdresult}}{\end{mdresult}\end{mdframed}}

\allowdisplaybreaks

\renewcommand{\qed}{\nobreak \ifvmode \relax \else
      \ifdim\lastskip<1.5em \hskip-\lastskip
      \hskip1.5em plus0em minus0.5em \fi \nobreak
      \vrule height0.75em width0.5em depth0.25em\fi}

\usepackage[T1]{fontenc}
\usepackage[utf8]{inputenc}

\newcommand{\ourinfo}[1]{Department of Computer and Information Science, University of Pennsylvania. Supported in part by National Science Foundation grants CCF-1552909, CCF-1617851, and IIS-1447470.  \newline\noindent Email: \texttt{#1}.}

\newcommand{\toShrink}{-.20cm}
\newcommand{\toShrinkEnu}{-.2cm}

\newcommand{\Mstar}{\ensuremath{M^{\star}}}

\newcommand{\Ot}{\ensuremath{\widetilde{O}}}
\newcommand{\eps}{\ensuremath{\varepsilon}}

\newcommand{\bracket}[1]{\left[#1\right]}
\newcommand{\paren}[1]{\ensuremath{\left(#1\right)}\xspace}
\newcommand{\card}[1]{\left\vert{#1}\right\vert}

\newcommand{\ceil}[1]{{\left\lceil{#1}\right\rceil}}
\newcommand{\floor}[1]{{\left\lfloor{#1}\right\rfloor}}

\newcommand{\set}[1]{\ensuremath{\left\{ #1 \right\}}}

\newcommand{\polylog}{\mbox{\rm  polylog}}

\newcommand{\opt}{\textnormal{\ensuremath{\mbox{opt}}}\xspace}

\DeclareMathOperator*{\Exp}{\ensuremath{{\mathbb{E}}}}
\DeclareMathOperator*{\Prob}{\ensuremath{\textnormal{Pr}}}
\renewcommand{\Pr}{\Prob}

\newcommand{\Ex}{\Exp}

\newcommand{\event}[1]{\ensuremath{{\sf E}_{#1}}}

\newenvironment{tbox}{\begin{tcolorbox}[
		enlarge top by=5pt,
		enlarge bottom by=5pt,
		 boxsep=0pt,
                  left=4pt,
                  right=4pt,
                  top=10pt,
                  arc=0pt,
                  boxrule=1pt,toprule=1pt,
                  colback=white
                  ]
	}
{\end{tcolorbox}}

\newcommand{\dist}{\ensuremath{\mathcal{D}}}

\newcommand{\Prot}{\ensuremath{\Pi}}

\newcommand{\bA}{\bm{A}}

\newcommand{\bB}{\ensuremath{\bm{B}}}

\renewcommand{\event}{\mathcal{E}}

\newcommand{\Matching}{\ensuremath{\textnormal{\textsf{Matching}}}}

\newcommand{\istar}{\ensuremath{i^{\star}}}

\newcommand{\player}[1]{\ensuremath{P^{(#1)}}}
\newcommand{\PS}[1]{\player{#1}}

\newcommand{\textbox}[2]{
{
\begin{tbox}
\textbf{#1}
{#2}
\end{tbox}
}
}

\newcommand{\vccs}{\ensuremath{\textnormal{\textsf{VC-Coreset}}}\xspace}

\newcommand{\bO}{\overline{O}}

\newcommand{\Gi}[1]{\ensuremath{G^{(#1)}}}
\newcommand{\Gii}{\Gi{i}}

\newcommand{\Vi}[1]{\ensuremath{V^{(#1)}}}
\newcommand{\Vii}{\Vi{i}}

\newcommand{\Ei}[1]{\ensuremath{E^{(#1)}}}
\newcommand{\Eii}{\Ei{i}}

\newcommand{\Vcs}{\ensuremath{V_{\textnormal{\textsf{cs}}}}}
\newcommand{\Vcsi}[1]{\Vcs^{(#1)}}
\newcommand{\Vcsii}{\Vcsi{i}}

\newcommand{\Vvc}{\ensuremath{O^{\star}}}
\newcommand{\bVvc}{\ensuremath{\overline{\Vvc}}}

\newcommand{\vc}{\ensuremath{\textnormal{\textsf{VC}}}}

\newcommand{\GreedyMatch}{\ensuremath{\textnormal{\textsf{GreedyMatch}}}\xspace}

\renewcommand{\event}[1]{\mathcal{E}(#1)}

\newcommand{\Ms}[1]{\ensuremath{M}^{(#1)}}

\newcommand{\mm}{\ensuremath{\textnormal{\textsf{MM}}}}

\newcommand{\New}{\ensuremath{\textnormal{\textsf{new}}}}
\newcommand{\Old}{\ensuremath{\textnormal{\textsf{old}}}}

\newcommand{\Vnew}{\ensuremath{V_{\New}}}
\newcommand{\Vold}{\ensuremath{V_{\Old}}}

\newcommand{\Mnew}{\ensuremath{M_{\New}}}

\newcommand{\muold}{\ensuremath{\mu_{\Old}}}
\newcommand{\Mold}{\ensuremath{M_{\Old}}}

\newcommand{\Enew}{E_{\New}}
\newcommand{\Eold}{E_{\Old}}

\newcommand{\Eiold}{E^{i}_{\Old}}

\renewcommand{\event}[1]{\ensuremath{\mathcal{E}\paren{#1}}}

\newcommand{\Mstarli}{\ensuremath{M^{\star< i}}}
\newcommand{\Mstari}{\ensuremath{M^{\star(i)}}}

\renewcommand{\bA}{\ensuremath{\overline{A}}}
\renewcommand{\bB}{\ensuremath{\overline{B}}}

\newcommand{\Eab}{\ensuremath{E_{AB}}}
\newcommand{\Ebab}{\ensuremath{E_{\overline{AB}}}}

\newcommand{\Mi}[1]{\ensuremath{M^{(#1)}}}
\newcommand{\Mii}{\Mi{i}}

\newcommand{\Ps}[1]{\ensuremath{P^{(#1)}}}

\newcommand{\distMatch}{\ensuremath{\mathcal{D}_{\Matching}}}
\newcommand{\distMD}{\ensuremath{\mathcal{D}_{\textnormal{\textsf{MR}}}}}

\newcommand{\vecB}{\ensuremath{\mathbf{B}}}

\newcommand{\FC}{\ensuremath{\mathcal{F}}}

\renewcommand{\event}{\ensuremath{\mathcal{E}}\xspace}

\newcommand{\estar}{\ensuremath{e^{\star}}}

\newcommand{\VertexCollection}{\ensuremath{\textnormal{\textsf{VertexCollection}}}\xspace}

\newcommand{\MatchingRecovery}{\ensuremath{\textnormal{\textsf{MatchingRecovery}}}\xspace}

\newcommand{\MA}{\ensuremath{M_{\textnormal{\textsf{Alice}}}}}
\newcommand{\EB}{\ensuremath{E_{\textnormal{\textsf{Bob}}}}}

\newcommand{\ProtMatching}{\ensuremath{\Prot}_{\textnormal{\textsf{Matching}}}\xspace}

\newcommand{\XM}{\ensuremath{X_M}}
\newcommand{\XbM}{\ensuremath{X_{\overline{M}}}}
\newcommand{\YM}{\ensuremath{Y_M}}
\newcommand{\YbM}{\ensuremath{Y_{\overline{M}}}}

\newcommand{\distVS}{\ensuremath{\dist_{\textnormal{\textsf{HVP}}}}}

\newcommand{\distVC}{\ensuremath{\dist_{\textnormal{\textsf{VC}}}}}

\newcommand{\vstar}{\ensuremath{v^{\star}}}

\newcommand{\EA}{\ensuremath{E_{A}}}

\newcommand{\ustar}{\ensuremath{u^{\star}}}

\newcommand{\Epi}[1]{\ensuremath{\widehat{E}^{(#1)}}}

\newcommand{\errs}{\ensuremath{\textnormal{errs}\xspace}}

\newcommand{\Dzi}{\ensuremath{D^{(i)}_0}}
\newcommand{\Doi}{\ensuremath{D^{(i)}_1}}
\newcommand{\Dti}{\ensuremath{D^{(i)}_{\geq 2}}}

\newcommand{\Dmi}{\ensuremath{D^{(i)}_{\leq 1}}}

\newcommand{\cVS}{\ensuremath{C_{\textnormal{\textsf{HVP}}}}}

\newcommand{\VertexSeeking}{\ensuremath{\textnormal{\textsf{HVP}}}\xspace}

\newcommand{\ProtVC}{\ensuremath{\Prot_{\textnormal{\textsf{VC}}}}}

\newcommand{\ProtVS}{\ensuremath{\Prot_{\textnormal{\textsf{HVP}}}}}

\newcommand{\distDisj}{\ensuremath{\dist_{\textnormal{\textsf{Disj}}}}}

\title{Randomized Composable Coresets for Matching and Vertex Cover}
\author{Sepehr Assadi\thanks{\ourinfo{\{sassadi,sanjeev\}@cis.upenn.edu}}  \and Sanjeev Khanna\footnotemark[1]}

\date{}

\begin{document}
\maketitle

\thispagestyle{empty}
\begin{abstract} 
A common approach for designing scalable algorithms for massive data sets is to distribute the computation across, say $k$, machines and process the data using limited communication between them. A particularly 
appealing framework here is the simultaneous communication model whereby each machine constructs a small representative summary of its own data and one obtains an approximate/exact solution from the union of the representative summaries.
If the representative summaries needed for a problem are small, then this results in a \emph{communication-efficient} and \emph{round-optimal} (requiring essentially no interaction between the machines) protocol. Some well-known examples of techniques for creating summaries include sampling, linear sketching, and composable coresets. These techniques have been successfully used to design communication efficient solutions for many fundamental graph problems. However, two prominent problems are notably absent from the list of successes, namely,
the \emph{maximum matching} problem and the \emph{minimum vertex cover} problem. Indeed, it was shown recently that for both these problems, even achieving a modest approximation factor of $\polylog{(n)}$ requires using representative summaries of size $\widetilde{\Omega}(n^2)$ i.e. essentially no better summary exists than each machine simply sending its entire input graph.
	
The main insight of our work is that the intractability of matching and vertex cover in the simultaneous communication model is inherently connected to an \emph{adversarial} partitioning of the underlying graph across machines. We show that when the underlying graph is randomly partitioned across machines, both these problems admit \emph{randomized composable coresets} of size $\widetilde{O}(n)$ that yield an $\widetilde{O}(1)$-approximate solution\footnote{Here and throughout the paper, we use $\Ot(\cdot)$ notation to suppress $\polylog{(n)}$ factors, where $n$ is the number of vertices in the graph.}. In other words, a small \emph{subgraph} of the input graph at each machine can be identified as its representative summary and the final answer then is obtained by simply running any maximum matching or minimum vertex cover algorithm on these combined subgraphs. 
This results in an $\widetilde{O}(1)$-approximation \emph{simultaneous} protocol for these problems with
$\Ot(nk)$ total communication when the input is randomly partitioned across $k$ machines. We also prove our results are optimal in a very strong sense: we not only rule out existence of smaller randomized composable coresets for these problems but in fact show that our $\Ot(nk)$ bound for total communication is optimal for {\em any} simultaneous communication protocol (i.e. not only for randomized coresets) for these two problems. Finally, by a standard application of composable coresets, our results also imply MapReduce algorithms with the same approximation guarantee in one or two rounds of communication, improving the previous best known round complexity for these problems. 

\end{abstract}

\clearpage
\setcounter{page}{1}

\newcommand{\Alg}{\ensuremath{\textnormal{\textsf{ALG}}}\xspace}
\newcommand{\EE}{\ensuremath{\mathcal{E}}}
\renewcommand{\tilde}{\widetilde}

\section{Introduction} \label{INTRO}

Recent years have witnessed tremendous algorithmic advances for efficient processing of massive data sets. 
A common approach for designing scalable algorithms for massive data sets is to distribute the computation across machines that are interconnected via a communication network. 
These machines can then jointly compute a function on the union of their inputs by exchanging messages. Two main measures of efficiency in this setting are 
the \emph{communication cost} and the \emph{round complexity}; we shall formally define these terms in details later in the paper but for the purpose of this section, communication 
cost measures the total number of bits exchanged by all machines and round complexity measures the number of rounds of interaction between them. 

An important and widely studied framework here is the \emph{simultaneous} communication model whereby each machine constructs a small representative summary of its own data and one obtains a solution for the desired problem 
from the union of the representative summary of combined pieces. 
The appeal of this framework lies in the simple fact that the \emph{simultaneous protocols} are inherently \emph{round-optimal}; they perform in only one round of interaction. The only measure that remains to be optimized is the 
communication cost -- this is now determined by the size of the summary created by each machine. An understanding of the communication cost for a problem in the simultaneous model turns out to have value in other models of 
computation as well. For instance, a lower bound on the maximum communication needed by any machine implies a matching lower bound on the space complexity of the same problem in dynamic streams~\cite{LiNW14,AiHLW16}. 

Two particularly successful techniques for designing small summaries for simultaneous protocols are \emph{linear sketches} and \emph{composable coresets}. 
Linear sketching technique corresponds to taking a {linear projection} of the input data as its representative summary. The ``linearity'' of the sketches is then used to obtain a
sketch of the combined pieces from which the final solution can be extracted.  
There has been a considerable amount of work in designing linear sketches
for graph problems in recent years~\cite{AhnGM12Linear,AGM12,KapralovLMMS14,AssadiKLY16,ChitnisCEHMMV16,BulteauFKP16,KW14,BhattacharyaHNT15,McGregorTVV15}. 
Coresets are subgraphs (in general, subsets of the input) that suitably preserve properties of a given graph, and they are said to be composable if the union of coresets for a collection of graphs yields a coreset for the union of the graphs. Composable coresets have also been studied extensively recently~\cite{BadanidiyuruMKK14,BalcanEL13,BateniBLM14,IndykMMM14,MirzasoleimanKSK13,MirrokniZ15}, and indeed
several graph problems admit natural composable coresets; for instance, connectivity, cut sparsifiers, and spanners (see~\cite{M14}, Section 2.2; the ``merge and reduce'' approach). 
Successful applications of these two techniques has yielded $\tilde{O}(n)$ size summaries for many graph problems (see further related work in Section~\ref{sec:related}). 
However, two prominent problems are notably absent from the list of successes, namely, the {\em maximum matching} problem and the {\em minimum vertex cover} problem. Indeed, it was shown recently~\cite{AssadiKLY16} that 
both matching and vertex cover require summaries of size $n^{2-o(1)}$ for even computing a $\polylog{(n)}$-approximate solution\footnote{The
authors in~\cite{AssadiKLY16} only showed the inapproximability result for the matching problem. However, a simple modification of their result proves an identical lower bound for the vertex cover problem 
as well.}. 

This state-of-affairs is the starting point for our work, namely, intractability of matching and vertex cover in the simultaneous communication model. Our main insight is that a
natural \emph{data oblivious partitioning scheme} completely alters this landscape: 
both problems admit $\tilde{O}(1)$-approximate composable coresets of size $\tilde{O}(n)$ provided the edges of the graph are randomly partitioned across the machines. The idea that random partitioning of 
data can help in distributed computation was nicely illustrated in the recent work of~\cite{MirrokniZ15} on maximizing submodular functions. Our work can be seen as the first illustration of this idea in the domain of graph algorithms.
The applicability of this idea to graph theoretic problems has been cast as an open problem in~\cite{MirrokniZ15}. 

\paragraph{Randomized Composable Coresets} We follow the notation of~\cite{MirrokniZ15} with a slight modification to adapt to our application in graphs. 
Let $E$ be an edge-set of a graph $G(V,E)$; we say that a partition $\set{\Ei{1},\ldots,\Ei{k}}$ of the edges $E$ is a \emph{random $k$-partitioning}  iff 
the sets are constructed by assigning each edge in $E$ independently to a set $\Ei{i}$ chosen uniformly at random. A random partitioning of 
the edges naturally defines partitioning the graph $G(V,E)$ into $k$ graphs $\Gi{1},\ldots,\Gi{k}$ whereby $\Gii := G(V,\Eii)$ for any $i \in [k]$, and hence we use random partitioning for 
both the edge-set and the input graph interchangeably. 

\begin{Definition}[Randomized Composable Coresets~\cite{MirrokniZ15}]
	For a graph-theoretic problem $P$, consider an algorithm \Alg that given any graph $G(V,E)$, outputs a 
	subgraph $\Alg(G) \subseteq G$ with at most $s$ edges. 
	Let ${\Gi{1},\ldots,\Gi{k}}$ be a \emph{random $k$-partitioning} of a graph $G$. 
	We say that \Alg outputs an \emph{$\alpha$-approximation randomized composable core-set} of \emph{size} $s$ for $P$ if
	${P\paren{\Alg(\Gi{1}) \cup \ldots \cup \Alg(\Gi{k})}}$ is an $\alpha$-approximation for $P(G)$ w.h.p., where the probability is taken over the random choice of the $k$-partitioning. 
	For brevity, we use randomized coresets to refer to randomized composable coresets.
\end{Definition}

We further augment this definition by allowing the coresets to also contain a \emph{fixed solution} to be \emph{directly} added to the final solution of the composed coresets. 
In this case, size of the coreset is measured both in the number of edges in the output subgraph plus the number of vertices and edges picked by the fixed solution (this is mostly relevant for our coreset for the vertex cover problem). 

\subsection{Our Results}\label{sec:results}

We show existence of randomized composable coresets for matching and vertex cover.
\begin{result}\label{res:alg}
	There exist randomized coresets of size $\Ot(n)$ that w.h.p. (over the random partitioning of the input) give an $O(1)$-approximation for maximum matching, and 
an $O(\log{n})$-approximation for minimum vertex cover.
\end{result}
 
In contrast to the above result, when the graph is {\em adversarially} partitioned, the results of~\cite{AssadiKLY16} show that the best approximation ratio conceivable for these problems in $\Ot(n)$ space is only $\Theta(n^{1/3})$.
We further remark that Result~\ref{res:alg} can also be extended to the weighted version of the problems. Using the Crouch-Stubbs technique~\cite{CS14} one can extend our result to achieve
a coreset for weighted matching (with a factor $2$ loss in approximation and extra $O(\log{n})$ term in the space). Similar ideas of ``grouping by weight'' of edges can also be used to extend our coreset for
weighted vertex cover with an $O(\log{n})$ factor loss in approximation and space; we omit the details.

The $\Ot(n)$ space bound achieved by our coresets above is considered a ``sweet spot'' for graph streaming algorithms~\cite{muthukrishnan2005data,FKMSZ05} as many fundamental problems are provably intractable in $o(n)$ space (sometimes not enough to even store the answer) while admit efficient solutions in $\Ot(n)$ space. However, in the simultaneous model, these considerations imply only that the total size of all $k$ coresets must be 
${\Omega}(n)$, leaving open the possibility that coreset output by each machine may be as small as $\Ot(n/k)$ in size (similar in spirit to coresets of~\cite{MirrokniZ15}). 
Our next result rules out this possibility and proves the optimality of our coresets size.
 
\begin{result}\label{res:lb-coreset}
Any $\alpha$-approximation randomized coreset for the matching problem must have size
$\Omega(n/\alpha^2)$, and any $\alpha$-approximation randomized coreset for the vertex cover problem must have size
$\Omega(n/\alpha)$.
\end{result}

We now elaborate on some applications of our results. 

\paragraph{Distributed Computation} We use the following distributed computation model in this paper, referred to as the \emph{coordinator model} (see, e.g.,~\cite{PhillipsVZ12}). The input is distributed across $k$ machines. 
There is also an additional party called the \emph{coordinator} who receives no input. The machines are allowed to only communicate with the coordinator, not with each other. A protocol
in this model is called a \emph{simultaneous} protocol iff the machines simultaneously send a message to the coordinator and the coordinator then outputs the answer with no further interaction. 
\emph{Communication cost} of a protocol in this model is the total number of bits communicated by all parties. 

Result~\ref{res:alg} can also be used to design simultaneous protocols for matching and vertex cover with $\Ot(nk)$ total communication and the same 
approximation guarantee stated in Result~\ref{res:alg} in the case the input is partitioned randomly across $k$ machines. Indeed, each machine only needs to compute a coreset
of its input, sends it to the coordinator, and coordinator computes an exact maximum matching or a $2$-approximate minimum vertex cover on the union of the  coresets. 
We further prove that the communication cost of theses protocols are essentially optimal.  

\begin{result}\label{res:lb-dist}
	Any $\alpha$-approximation simultaneous protocol for the maximum matching problem, resp. the vertex cover problem,
	requires total communication of $\Omega(nk/\alpha^2)$ bits, resp. $\Omega(nk/\alpha)$ bits, \emph{even} when the input 
	is \emph{partitioned randomly} across the machines.
\end{result}

Result~\ref{res:lb-dist} is a strengthening of Result~\ref{res:lb-coreset}; it rules out \emph{any} representative summary 
(not necessarily a randomized coreset) of size $o(n/\alpha^2)$ (resp. $o(n/\alpha)$) that can be used for $\alpha$-approximation of matching (resp. vertex cover) when the input is partitioned randomly.

For the matching problem, it was shown previously in~\cite{HuangRVZ15} that when the input is adversarially partitioned in the coordinator model, any protocol (not necessarily simultaneous) requires $\Omega(nk/\alpha^2)$ bits of 
communication to achieve an $\alpha$-approximation of the maximum matching. Result~\ref{res:lb-dist} extends this to the case of \emph{randomly partitioned} inputs albeit only for simultaneous protocols.

\paragraph{MapReduce Framework} We show how to use our randomized coresets to obtain improved MapReduce algorithms for matching and vertex cover in the MapReduce computation model
formally introduced in~\cite{KSV10,LMSV11}. 
Let $k = \sqrt{n}$ be the number of machines, each with a memory of $\Ot(n\sqrt{n})$; we show that \emph{two} rounds of MapReduce suffice to obtain 
an $O(1)$-approximation for matching and $O(\log{n})$-approximation for vertex cover. In the first round, each machine randomly partitions the edges assigned to it across the $k$ machines; this results in a 
random $k$-partitioning of the graph across the machines. In the second round, each machine sends a randomized composable coreset of its input to a designated central machine $M$; as there are 
$k = \sqrt{n}$ machines and each machine is sending $\Ot(n)$ size coreset, the input received by $M$ is of size $\Ot(n\sqrt{n})$ and hence can be stored entirely on that machine. Finally, $M$
computes the answer by combining the coresets (similar to the case in the coordinator model). Note that if the input was distributed randomly in the first place, we could have implemented this algorithm in only 
one round of MapReduce (see~\cite{MirrokniZ15} for details on when this assumption applies). 

Our MapReduce algorithm outperforms the previous algorithms of~\cite{LMSV11} for matching and vertex cover in terms of the number of rounds it uses, albeit with a larger approximation guarantee. In 
particular,~\cite{LMSV11} achieved a $2$-approximation to both matching and vertex cover in $6$ rounds of MapReduce when using similar space as ours on each machine (the number of rounds of this algorithm is always
at least $3$ even if we allow $\Ot(n^{5/3})$ space per each machine). 
The improvement on the number of rounds is significant in this context; the transition between different rounds in a MapReduce computation is usually the dominant cost 
of the computation~\cite{LMSV11} and hence, minimizing the number of rounds is an important goal in the MapReduce framework.

\subsection{Our Techniques}\label{sec:techniques}

\paragraph{Randomized Coreset for Matching} Greedy and Local search algorithms are the typical choices for composable coresets (see, e.g.,~\cite{IndykMMM14,MirrokniZ15}). It is then natural
to consider the greedy algorithm for the maximum matching problem as a randomized coreset: the one that computes a \emph{maximal matching}. However, 
one can easily show that this choice of coreset performs poorly in general; there are simple instances in which choosing arbitrary maximal matching in the graph $\Gii$ results only in an $\Omega(k)$-approximation. 

Somewhat surprisingly, we show that a simple change in strategy results in an efficient randomized coreset: \emph{any} \emph{maximum matching} of the graph $\Gii$
 can be used as an $O(1)$-approximate randomized coreset for the 
maximum matching problem. Unlike the previous work in~\cite{MirrokniZ15,IndykMMM14} that relied on analyzing a specific algorithm (or a specific family of algorithms) for constructing a coreset, we prove this result
by exploiting structural properties of the maximum matching (i.e., the optimal solution) directly, independent of the algorithm that computes it. As a consequence, our coreset construction requires no prior coordination  (such as consistent tie-breaking rules used in~\cite{MirrokniZ15}) between the machines and in fact each machine can use a different algorithm for computing the maximum matching required by the coreset. 

\paragraph{Randomized Coreset for Vertex Cover} In the light of our coreset for the matching problem, one might wonder whether a minimum vertex cover of a graph can also be used
as its randomized coreset. However, it is easy to show that the answer is negative here -- there are simple instances (e.g., a star on $k$ vertices) on which this leads to an $\Omega(k)$ approximation ratio. Indeed, the \emph{feasibility constraint} in the vertex cover problem depends heavily on the input graph
as a whole and not only the coreset computed by each machine, unlike the case for matching and in fact most problems that admit a composable coreset~\cite{BalcanEL13,IndykMMM14,MirrokniZ15}. This suggests the 
necessity of using edges in the coreset to \emph{certify} the feasibility of the answer. On the other hand, only sending edges seems too restrictive: a vertex of degree $n-1$ can safely be assumed to be in an optimal 
vertex cover, but to {certify} this, one needs to essentially communicate $\Omega(n)$ edges. This naturally motivates a slightly more general notion of coresets -- the coreset contains both subsets of vertices (to be always included
in the final vertex cover) and edges (to guide the choice of additional vertices in the vertex cover). 

To obtain a randomized coreset for vertex cover, we employ an iterative ``peeling'' process where we remove the vertices with the highest residual degree in each iteration (and add them to the final vertex cover) and continue until 
the residual graph is sufficiently sparse, in which case we can return this subgraph as the coreset. The process itself is a modification of the algorithm by Parnas and Ron~\cite{ParnasR07}; we point out that other modifications of this 
algorithm has also been used previously for matching and vertex cover~\cite{OnakR10,KapralovKS14,BhattacharyaHI15}. 

However, to employ this algorithm as a coreset we need to argue that the set of vertices peeled across different machines is not too large as these vertices are added directly to the final vertex cover. The intuition behind
this is that random partitioning of edges in the graph should result in vertices to have essentially the same degree across the machines and hence each machine should peel the same set of vertices in each iteration. But this intuition 
runs into a technical difficulty:
the peeling process is quite sensitive to the exact degree of vertices and even slight changes in degree results in moving vertices between different iterations that potentially leads to a cascading effect. To address this, 
we design a \emph{hypothetical} peeling process (which is aware of the actual minimum vertex cover in $G$) and show that the our actual peeling process is 
in fact ``sandwiched'' between two application of this peeling process with different degree threshold for peeling vertices. We then use this to argue that the set of all vertices peeled across the machines are always
contained in the solution of the hypothetical peeling process which in turn can be shown to be a relatively small set. 

\paragraph{Lower Bounds for Randomized Coresets.} 
Our lower bound results for randomized coresets for matching are based on the following simple distribution: the input graph consists of union of two bipartite graphs, one of which is a random $k$-regular
graph $G_1$ with $n/2\alpha$ vertices on each side while the other graph $G_2$ is a perfect matching of size $n - n/2\alpha$. Thus the input graph almost certainly contains a matching of size $n - o(n)$ and
any $\alpha$-approximate solution must collect $\Omega(n/\alpha)$ edges from $G_2$ overall i.e. $\Omega(n/\alpha k)$ edges from $G_2$ from each machine on average. After random partitioning, the input given to each machine 
is essentially a matching of size $n/2\alpha$ from $G_1$ and a matching of size roughly $n/k$ from $G_2$. The local information at each machine is not sufficient to differentiate between edges of $G_1$ and $G_2$, and thus any 
coreset that aims to include $\Omega(n/\alpha k)$ edges from $G_2$, can not reduce the input size by more than a factor of $\alpha$. Somewhat similar ideas can also be shown to work for the vertex cover problem.

\paragraph{Communication Complexity Lower Bounds} We briefly highlight the ideas used in obtaining the lower bounds described in Result~\ref{res:lb-dist}. We will focus on the vertex cover problem to describe our techniques. Our 
lower bound result is based on analyzing (a variant of) the following distribution: the input graph $G(L,R,E)$ consists of a bipartite graph $G_1$ 
plus a single edge $\estar$. $G_1$ is a graph on $n/2\alpha$ vertices $L_1 \subseteq L$, each connected to $k$ random neighbors in $R$, and $\estar$ is an edge chosen uniformly at random between $L \setminus L_1$ and $R$. 
This way $G$ admits a minimum vertex cover of size at most $n/2\alpha+1$. However, when this graph is randomly partitioned, the input
to each machine is essentially a matching of size $n/2\alpha$ chosen from the graph $G_1$ with possibly one more edge $\estar$ (in exactly one machine chosen uniformly at random). 
The local information at the machine receiving the edge $\estar$ is not sufficient to differentiate between the edges
of $G_1$ and $\estar$ and thus if the message sent by this machine is much smaller than its input size (i.e., $o(n/\alpha)$ bits), it most likely does not ``convey enough information'' to the coordinator about the identity of $\estar$.
This in turn forces the coordinator to use more than $n/2$ vertices in order to cover $\estar$, resulting in an approximation factor larger than $\alpha$.

Making this intuition precise is complicated by the fact that the input across the players are highly correlated, and hence the message sent by one player, can also reveal extra information about the input of another (e.g. a relatively small communication from the players is enough for the coordinator to know the identity of entire $L_1$). To overcome this, we show that by conditioning on proper parts of the input, we can limit the correlation
in the input of players and then use the \emph{symmetrization} technique of~\cite{PhillipsVZ12} to reduce the simultaneous $k$-player vertex cover problem to a one-way two-player problem named the \emph{hidden
vertex problem} (\VertexSeeking). 
Loosely speaking, in \VertexSeeking, Alice and Bob are given two sets $S,T \subseteq [n]$, each of size $n/\alpha$, with the promise that $\card{S \setminus T} = 1$ and their goal is to find a set $C$ of size $o(n)$ which
contains the single element in $S \setminus T$.  We prove a lower bound of $\Omega(n/\alpha)$ bits for this problem using a subtle reduction from the well-known set disjointness problem. In this reduction, Alice and Bob 
use the protocol for \VertexSeeking on ``non-legal'' instances (i.e., the ones for which \VertexSeeking is not well-defined) to reduce the original disjointness instance between sets $A,B$ on a universe $[N]$
to a lopsided disjointness instance $(A,B')$ whereby $\card{B'} = o(N)$, and then solve this new instance in $o(N)$ communication (using the H{\aa}stad-Wigderson protocol~\cite{HastadW07}), contradicting the $\Omega(N)$
lower bound on the communication complexity of disjointness. 

The lower bound for the matching problem is also proven along similar lines (over the hard distribution mentioned earlier for this problem) using
 a careful combinatorial argument instead of the reduction from the disjointness problem. 

\subsection{Further Related Work}\label{sec:related}

Maximum matching and minimum vertex cover are among the most studied problems in the context of massive graphs including, in dynamic graphs~\cite{NeimanS13,Solomon16,BaswanaGS15,OnakR10}, sub-linear
algorithms~\cite{ParnasR07,HassidimKNO09,NguyenO08,OnakRRR12,YoshidaYI12}, streaming 
algorithms~{\cite{M05,FKMSZ05,EKS09,EpsteinLMS11,AhnGM12Linear,GoelKK12,KonradMM12,AGM12,AG13,GO13,Kapralov13,KapralovKS14,CS14,ChitnisCHM15,M14,AhnG15,EsfandiariHLMO15,Konrad15,AssadiKLY16,ChitnisCEHMMV16,McGregorV16,EsfandiariHM16,AssadiKL17,PazS17}, MapReduce computation~\cite{AhnGM12Linear,LMSV11}, and different distributed computation models~\cite{HuangRVZ15,AlonNRW15,DNO14,GO13}. 
Most relevant to our work are the linear sketches of~\cite{ChitnisCEHMMV16} for computing an \emph{exact} minimum vertex cover or maximum matching in $O(\opt^2)$ space ($\opt$ is the size of the solution), and linear
sketches of~\cite{AssadiKLY16,ChitnisCEHMMV16} for $\alpha$-approximating maximum matching in $\Ot(n^2/\alpha^3)$ space. These results are proven to be tight by~\cite{ChitnisCHM15}, and~\cite{AssadiKLY16}, respectively.
Finally,~\cite{AssadiKLY16} also studied the simultaneous communication complexity of bipartite matching in the vertex-partition model and proved that obtaining better than an $O(\sqrt{k})$-approximation in this model 
 requires strictly more than $\Ot(n)$ communication from each player (see~\cite{AssadiKLY16} for more details on this model). 

Coresets, composable coresets, and randomized composable coresets are respectively introduced in~\cite{AgarwalHV04},~\cite{IndykMMM14}, and~\cite{MirrokniZ15}. Composable coresets 
have been used previously in the context of nearest neighbor search~\cite{AbbarAIMV13}, diversity maximization~\cite{IndykMMM14}, clustering~\cite{BalcanEL13,BateniBLM14}, and
submodular maximization~\cite{IndykMMM14,MirrokniZ15,BadanidiyuruMKK14}. Moreover, while not particularly termed a composable coreset, the ``merge and reduce'' technique
in the graph streaming literature (see~\cite{M14}, Section~2.2) is identical to composable coresets. Similar ideas as randomized coreset for optimization problems has also been used in random arrival streams~\cite{KonradMM12,KapralovKS14}. Moreover, communication complexity lower bounds have also been studied 
previously under the random partitioning of the input~\cite{KapralovKS15,ChakrabartiCM08}. 
\section{Preliminaries} \label{PRELIM}

\paragraph{Notation.} For any integer $m$, $[m] := \set{1,\ldots,m}$. Let $G(V,E)$ be a graph; $\mm(G)$ denotes the maximum matching size in $G$ and $\vc(G)$ denotes the minimum vertex cover size. 
We assume that these quantities are $\omega(k\log{n})$\footnote{Otherwise, we can use the algorithm of~\cite{ChitnisCEHMMV16} to obtain \emph{exact} coresets of size $\Ot(k^2)$ as mentioned in Section~\ref{sec:related}.}. For a 
set $S \subseteq V$ and $v \in V$, $N_S(v) \subseteq S$ denotes the neighbors of $v$ in the set $S$. For an edge set $E' \subseteq E$, we use $V(E')$ to refer to vertices
 incident on $E'$. 
 
 \paragraph{Useful Concentration of Measure Inequalities.} We use the following standard version of Chernoff bound (see, e.g.,~\cite{ConcentrationBook}) throughout. 

\begin{proposition}[Chernoff bound]\label{prop:chernoff}
	Let $X_1,\ldots,X_n$ be independent random variables taking values in $[0,1]$ and let $X:= \sum_{i=1}^{n} X_i$. Then, 
	\[ \Pr\paren{\card{X - \Ex\bracket{X}} > t} \leq 2 \cdot \exp\paren{-\frac{2t^2}{n}} \]
\end{proposition}

We also need the {method of bounded differences} in our proofs. A function $f(x_1,\ldots,x_n)$ satisfies the \emph{Lipschitz property} with constant $d$, iff 
for all $i \in [n]$, $\card{f(a) - f(a')} \leq d$, whenever $a$ and $a'$ differ only in the $i$-th coordinate. 

\begin{proposition}[Method of bounded differences] \label{prop:bounded-differences}
	If $f$ satisfies the Lipschitz property with constant $d$ and $X_1,\ldots,X_n$ are independent random variables, then, 
	\[ \Pr\paren{\card{f(X) - \Ex\bracket{f(X)}} > t} \leq 2 \cdot \exp\paren{-\frac{2t^2}{n\cdot d^2}} \]
\end{proposition}

A proof of this proposition can be found in~\cite{ConcentrationBook} (see Section 5).

\paragraph{Communication Complexity} We prove our lower bounds for distributed protocols using the framework of communication complexity, and in particular in the \emph{multi-party simultaneous communication model} and 
the \emph{two-player one-way communication model}. 

Formally, in the multi-party simultaneous communication model, the input is partitioned across $k$ players $\Ps{1},\ldots,\Ps{k}$. All players have access to an infinite  
shared string of random bits, referred to as \emph{public randomness} (or \emph{public coins}). The goal is for the players to compute a specific function of the input by simultaneously sending a message to a central party called 
the coordinator (or the referee). The coordinator then needs to output
the answer using the messages received by the players. 
We refer to the case when the input is partitioned randomly as the \emph{random partition} model. 

In the {two-player one-way communication model}, the input is partitioned across two players, namely Alice and Bob. The players again have access to public randomness, and the goal is for Alice to send a single 
message to Bob, so that Bob can compute a function of the joint input. The \emph{communication cost} of a protocol in both models is the total length of the messages sent by the players. 
In Section~\ref{sec:cc-vertex-seeking}, we also consider general two-player communication model which allows a \emph{two-way} communication, i.e., both Alice and Bob can send messages to each other. 
We refer the reader to an excellent text by Kushilevitz and Nisan~\cite{KN97} for more details.

\section{Randomized Coresets for Matching and Vertex Cover}\label{UPPER}
We present our randomized composable coresets for matching and vertex cover in this section.
\subsection{An $O(1)$-Approximation Randomized Coreset for Matching} \label{sec:matching-coreset}

The following theorem formalizes Result~\ref{res:alg} for matching. 

\begin{theorem}\label{thm:matching}
	Any \emph{maximum matching} of a graph $G(V,E)$ is an $O(1)$-approximation randomized composable coreset of size $O(n)$ for the maximum matching problem. 
\end{theorem}

We remark that our main interest in Theorem~\ref{thm:matching} is to achieve \emph{some} constant approximation factor for randomized composable coresets of the matching problem and as such we 
did not optimize the constant in the approximation ratio. Nevertheless, our result already shows that the approximation ratio of this coreset is \emph{at most $9$} (in fact, with a bit more care, 
we can reduce this factor down to $8$; however, as this is not the main contribution of this paper, we omit the details). 

Let $G(V,E)$ be any graph and $\Gi{1},\ldots,\Gi{k}$ be a random $k$-partitioning of $G$. To prove Theorem~\ref{thm:matching}, we describe
a simple process for combining the maximum matchings (i.e., the coresets) of $\Gii$'s, and prove that this process results in a constant factor
approximation of the maximum matching of $G$. We remark that this process is only required for the analysis, i.e., to show that there exists a large matching in the union of coresets; in principle, any (approximation) 
algorithm for computing a maximum matching can be applied to obtain a large matching from the coresets. 

Consider the following greedy process for computing an approximate matching in $G(V,E)$: 
\textbox{$\GreedyMatch(G)$:}{
\begin{enumerate}
	\item Let $\Ms{0} := \emptyset$. For $i=1$ to $k$: 
	\item  Let $\Ms{i}$ be a \emph{maximal matching} obtained by adding to $\Ms{i-1}$ the edges in an \emph{arbitrary maximum matching} of $\Gii$ 
		that do not violate the matching property. 
	\item return $M:= \Ms{k}$. 
\end{enumerate}
}

\begin{lemma}\label{lem:greedy-match}
	$\GreedyMatch$ is an $O(1)$-approximation algorithm for the maximum matching problem w.h.p (over the randomness of the edge partitioning). 	
\end{lemma}

Before proving Lemma~\ref{lem:greedy-match}, we show that Theorem~\ref{thm:matching} easily follows from this lemma. 

\begin{proof}[Proof of Theorem~\ref{thm:matching}] 
Let $\Alg$ be any algorithm that given a graph $G(V,E)$, $\Alg(G)$ outputs an arbitrary maximum matching of $G$. It is 
immediate to see that to implement $\GreedyMatch$, we only need to compute a maximal matching on the output of \Alg on each graph $\Gii$ where $\Gii$'s form a random $k$-partitioning of $G$. 
Consequently, since $\GreedyMatch$ outputs an $O(1)$-approximate matching (by Lemma~\ref{lem:greedy-match}),  
the graph $H:= \Gi{1} \cup \ldots \cup \Gi{k}$ should contain an $O(1)$-approximate matching as well. We emphasize here that the use 
of \GreedyMatch for finding a large matching in $H$ is \emph{only} for the purpose of analysis. 
\end{proof}

In the rest of this section, we prove Lemma~\ref{lem:greedy-match}. 
Recall that $\mm(G)$ denotes the maximum matching size in the input graph $G$. Let $c > 0$ be a small constant to be determined later. 
To prove Lemma~\ref{lem:greedy-match}, we will show that ${\card{\Ms{k}}} \geq c \cdot \mm(G)$ w.h.p, where $\Ms{k}$ is the output of
$\GreedyMatch$. Notice that the matchings $\Ms{i}$ (for $i \in [k]$) constructed by \GreedyMatch are random variables depending on the random $k$-partitioning.

Our general approach for the proof of Lemma~\ref{lem:greedy-match} is as follows. Suppose at the beginning of the $i$-th step of \GreedyMatch, the matching $\Ms{i-1}$ is of size $o(\mm(G))$. It is easy to see that 
in this case, there is a matching of size $\Omega(\mm(G))$ in $G$ that is entirely incident on vertices of $G$ that are not matched by $\Ms{i-1}$. We can further show that in fact 
$\Omega(\mm(G)/k)$ edges of this matching are appearing in $\Gii$, \emph{even} when we condition on the assignment of the edges in the first $(i-1)$ graphs. The next step is then to argue
that the existence of these edges forces \emph{any} maximum matching of $\Gii$ to match $\Omega(\mm(G)/k)$ edges in $\Gii$ between the vertices that are not matched by $\Ms{i-1}$; these edges 
can always be added to the matching $\Ms{i-1}$ to form $\Ms{i}$. This ensures that while the maximal matching in \GreedyMatch is of size $o(\mm(G))$, we can increase its size by $\Omega(\mm(G)/k)$ edges in each of the first $k/3$ steps, hence obtaining a matching of size $\Omega(\mm(G))$ at the end. The following key lemma formalizes this argument. 

\begin{lemma}\label{lem:increase-step}
	For any $i \in [k/3]$, if ${\card{\Ms{i-1}}} \leq c \cdot \mm(G)$, then, w.p. $1-O(1/n)$, 
	\[ \card{\Ms{i}} \geq \card{\Ms{i-1}} + \paren{\frac{1-6c-o(1)}{k}}\cdot\mm(G)\]
\end{lemma}


To continue we define some notation. Let $\Mstar$ be an arbitrary maximum matching of $G$. For any $i \in [k]$, we define $\Mstarli$ as the part of $\Mstar$ assigned to the first $i-1$ graphs in the random $k$-partitioning, i.e., the graphs $\Gi{1},\ldots,\Gi{i-1}$. We have the following simple concentration result. 
\begin{claim}\label{clm:matching-concentration-first-i}
	W.p. $1-O({1}/{n})$, for any $i \in [k]$, $$\card{\Mstarli} \leq \paren{\frac{i-1+o(i)}{k}} \cdot \mm(G).$$ 
\end{claim}
\begin{proof}
	Fix an $i \in [k]$; each edge in $\Mstar$ is assigned to $\Gi{1},\ldots,\Gi{i-1}$, w.p. $(i-1)/k$, hence in expectation, size of $\Mstarli$ is $\frac{i-1}{k} \cdot \mm(G)$. The claim now follows from a standard application of 
	Chernoff bound (recall that, throughout the paper, we assume $\mm(G) = \omega(k \log{n})$). 
\end{proof}
We now prove Lemma~\ref{lem:increase-step}.
\begin{proof}[Proof of Lemma~\ref{lem:increase-step}]
	Fix an $i \in [k/3]$ and the set of edges for $\Ei{1},\ldots,\Ei{i-1}$; this also fixes the matching $\Ms{i-1}$ while the set of edges in $\Eii,\ldots,\Ei{k}$ together with the matching $\Ms{i}$ are still random variables. 
	We further assume that after fixing the edges in $\Ei{1},\ldots,\Ei{i-1}$,  $\card{\Mstarli} \leq \frac{i-1+o(i)}{k} \cdot \mm(G)$ which 
	happens w.p. $1-O(1/n)$ by Claim~\ref{clm:matching-concentration-first-i}. 
		
	We first define some notation. Let $\Vold$ be the set of vertices incident on $\Ms{i-1}$ and $\Vnew$ be the remaining vertices. Let $E^{\geq i}$ be the set of edges in $E \setminus \paren{\Ei{1} \cup \ldots \cup \Ei{i-1}}$. 
	We partition $E^{\geq i}$ into two parts: $(i)$ $\Eold$: the set of edges with \emph{at least one endpoint} in $\Vold$, and $(ii)$ $\Enew$: the set of edges \emph{incident entirely} on $\Vnew$. 
	Our goal is to show that w.h.p. \emph{any} maximum matching of $\Gii$ matches $\Omega(\mm(G)/k)$ vertices in $\Vnew$ to each other by using the edges in $\Enew$; the lemma
	then follows easily from this. 
	 
	Notice that the edges in the graph $\Gii$ are chosen by independently assigning each edge in $E^{\geq i}$ to $\Gii$ w.p. $1/(k-i+1)$\footnote{This is true even when we condition on the size of $\card{\Mstarli}$ since this 
	event does not depend on the choice of edges in $E^{\geq i}$.}. This independence allows us to treat the edges in $\Eold$ and $\Enew$
	separately; we can fix the set of sampled edges of $\Gii$ in $\Eold$ denoted by $\Eiold$ without changing the distribution of edges in $\Gii$ chosen from $\Enew$.  
	 Let $\muold:= \mm(G(V,\Eiold))$, i.e., the maximum number of edges that can be matched
	in $\Gii$ using only the edges in $\Eiold$.  In the following, we show that w.h.p., there exists a matching of size $\muold + \Omega(\mm(G)/k)$ in $\Gii$; by the definition of $\muold$, 
	this implies that \emph{any} maximum matching of $\Gii$ has to use at least $\Omega(\mm(G)/k)$ edges in $\Enew$, proving the lemma. 
	 
	Let $\Mold$ be any arbitrary maximum matching of size $\muold$ in $G(V,\Eiold)$. Let $\Vnew(\Mold)$ be the set of vertices in $\Vnew$ that are incident on $\Mold$. 
	We show that there is a large matching in $G(V,\Enew)$ that avoids $\Vnew(\Mold)$. 
	\begin{claim}\label{clm:vnew-avoid}
		There exists a matching in $G(V,\Enew)$ of size $\paren{\frac{k-i+1-o(i)}{k} - 4c} \cdot \mm(G)$ that avoids the vertices of $\Vnew(\Mold)$. 	
	\end{claim}
	\begin{proof}
	 	We first bound the size of $\Vnew(\Mold)$. Since any edge in $\Mold$ has at least one endpoint in $\Vold$, we have $\card{\Vnew(\Mold)} \leq \card{\Mold} \leq \card{\Vold}$. 
		By the assertion of the lemma, $\card{\Ms{i-1}} < c \cdot \mm(G)$, and hence $\card{\Vnew(\Mold)} \leq \card{\Vold} < 2c \cdot \mm(G)$. 
		
		Moreover, by the assumption that $\card{\Mstarli} \leq \frac{i-1+o(i)}{k} \cdot \mm(G)$, there is a matching $M$ of size $\frac{k-i+1-o(i)}{k} \cdot \mm(G)$ in the graph $G(V,E^{\geq i})$. 
		By removing the edges in $M$ that are either incident on $\Vold$ or $\Vnew(\Mold)$,  at most $4c \cdot \mm(G)$ edges are removed from $M$. Now the remaining matching is entirely contained 
		in $\Enew$ and also avoids $\Vnew(\Mold)$, hence proving the claim. 
	 \end{proof}
	
	We are now ready to finalize the proof. Let $\Mnew$ be the matching guaranteed by Claim~\ref{clm:vnew-avoid}. Each edge in this matching is chosen in $\Gii$ w.p. $1/(k-i+1)$ independent of the 
	other edges; hence, by Chernoff bound (and the assumption that $\mm(G) = \omega(k\log{n})$), there is a matching of size 
	\begin{align*}
		(1-o(1)) \cdot \paren{\frac{1}{k} - \frac{o(i)}{k(k-i+1)} - \frac{4c}{k-i+1}} \cdot \mm(G) \\
		\geq \paren{\frac{1-6c-o(1)}{k}}\cdot\mm(G) \tag{$i \leq k/3$}
	\end{align*}
	in the edges of $\Mnew$ that appear in $\Gii$. This matching can be directly added to the matching $\Mold$, implying the existence of a matching of size $\muold + \paren{\frac{1-6c-o(1)}{k}}\cdot\mm(G)$ in $\Gii$. 
	As argued before, this ensures that any maximum matching of $\Gii$ contains at least $\paren{\frac{1-6c-o(1)}{k}}\cdot\mm(G)$ edges in $\Enew$. 
	These edges can always be added to $\Ms{i-1}$ to form $\Ms{i}$, hence proving the lemma.
\end{proof}

Having proved Lemma~\ref{lem:increase-step}, we can easily conclude Lemma~\ref{lem:greedy-match}. 


\begin{proof}[Proof of Lemma~\ref{lem:greedy-match}]
	Recall that $M:= \Ms{k}$ is the output matching of \GreedyMatch. For the first $k/3$ steps of $\GreedyMatch$, if at any step we obtained a matching of size $c \cdot \mm(G)$, then we are already done. 
	Otherwise, at each step, by Lemma~\ref{lem:increase-step}, w.p. $1-O(1/n)$, we increase the size of the maximal matching by $\paren{\frac{1-6c-o(1)}{k}}\cdot\mm(G)$ edges; consequently, by taking a union bound on the $k/3$ 
	steps, w.p. $1-o(1)$, the size of the maximal matching would be $\paren{\frac{1-6c-o(1)}{3}}\cdot\mm(G)$. By picking $c = 1/9$, we ensure that in either case, the matching computed by $\GreedyMatch$ is of 
	size at least $\mm(G)/9 - o(\mm(G))$, proving the lemma.    
\end{proof}



\subsection{An $O(\log{n})$-Approximation Randomized Coreset For Vertex Cover} \label{sec:vc-coreset}

The following theorem formalizes Result~\ref{res:alg} for vertex cover. 

\begin{theorem}\label{thm:vc}
	There exists an $O(\log{n})$-approximation randomized composable coreset of size $O(n\log{n})$ for the vertex cover problem. 
\end{theorem}

Let $G(V,E)$ be a graph  and $\Gi{1},\ldots,\Gi{k}$ be a random $k$-partitioning of $G$; we propose the following coreset for computing an approximate vertex cover of $G$. This coreset construction is 
a modification of the algorithm for vertex cover first proposed by~\cite{ParnasR07}.  

\textbox{$\vccs(\Gii)$. \textnormal{An algorithm for computing a composable coreset of each $\Gii$.}}{ 
\begin{enumerate}
	\item Let $\Delta$ be the smallest integer such that $n / (k \cdot 2^{\Delta}) \leq 4\log{n}$ and define $\Gii_1 := \Gii$. 
	\item For $j = 1$ to $\Delta-1$, let:
	\begin{align*} \Vii_j &:= \set{\text{vertices of degree} \geq n / (k \cdot 2^{j+1}) \text{ in $\Gii_j$}} \\
		  \Gii_{j+1} &:= \Gii_j \setminus \Vii_j. 
	\end{align*}
	\item Return $\Vcsii:= \bigcup_{j=1}^{\Delta-1} \Vii_j$ as a \emph{fixed solution} plus the graph $\Gii_{\Delta}$ as the coreset. 
\end{enumerate}
}

In $\vccs$ we allow the coreset to, in addition to returning a subgraph, identify a set of vertices (i.e., $\Vcsii$) to be added directly to the final vertex cover. 
In other words, to compute a vertex cover of the graph $G$, we compute a vertex cover of the graph $\bigcup_{i=1}^{k} \Gii_{\Delta}$ and return it together with the vertices $\bigcup_{i=1}^{k} \Vcsii$. It is easy to see that this set of vertices indeed forms a vertex cover of $G$: any edge in $G$ that belongs to $\Gii$ is either incident on some $\Vii_j$, and hence is covered by $\Vii_j$, or is present in $\Gii_{\Delta}$, and hence is covered by the
vertex cover of $\Gii_{\Delta}$. 

In the remainder of this section, we bound the approximation ratio of this coreset. To do this, we need to prove that $\card{\bigcup_{i=1}^{k} \Vcsii} = O(\log{n}) \cdot \vc(G)$. The bound 
on the approximation ratio then follows as the vertex cover of $\bigcup_{i=1}^{k} \Gii_\Delta$ can be computed
to within a factor of $2$. 

It is easy to prove (and follows from~\cite{ParnasR07}) that the set of vertices $\Vcsii$ is of size $O(\log{n}) \cdot \vc(G)$; however, using this fact directly to bound the size of $\bigcup_{i=1}^{k} \Vcsii$ 
only implies an approximation ratio of $O(k \log{n})$ which is far worse than our goal of achieving an $O(\log{n})$-approximation. In order to obtain the $O(\log{n})$ bound, we need to argue that 
not only each set $\Vcsii$ is relatively small, but also that these sets are all intersecting in many vertices. In order to do so, we introduce a hypothetical algorithm (similar to \vccs) 
on the graph $G$ and argue that the set $\Vcsii$ output by $\vccs(\Gii)$ is, with high probability, a subset of the output of this hypothetical algorithm. This allows us to then bound the size of the 
union of the sets $\Vcsii$ for $i \in [k]$. 

Let $\Vvc$ denote the set of vertices in an arbitrary optimum vertex cover of $G$ and $\bVvc := V \setminus \Vvc$. Consider the following process on the original graph $G$ (defined only for analysis): 

\textbox{\vspace{-0.5cm}}{
\begin{enumerate}
	\item Let $G_1$ be the bipartite graph obtained from $G$ by removing edges between vertices in $\Vvc$. 
	\item For $j = 1$ to $t:=\ceil{\log{n}}$, let: 
	\begin{align*} 
	O_j &:= \set{\text{vertices in $\Vvc$ of degree} \geq n / 2^{j} \text{ in $G_j$}} \\ 
	\bO_j &:= \set{\text{vertices in $\bVvc$ of  degree} \geq n / 2^{j+2} \text{ in $G_j$}} \\ 
	G_{j+1} &:= G_j \setminus (O_j \cup \bO_j). 
	\end{align*} 

\end{enumerate}
}

We first prove that the sets $O_j$'s and $\bO_j$'s in this process form an $O(\log{n})$ approximation of the minimum vertex cover of $G$ and then show that 
$\vccs(\Gii)$ (for any $i \in [k]$) is \emph{mimicking} this hypothetical process in a sense that the set $\Vcsii$ is essentially \emph{contained} in the union of the sets $O_j$'s and $\bO_j$'s. 

\begin{lemma}\label{lem:vc-analysis-ratio}
	$\card{\bigcup_{j=1}^{t} O_j \cup \bO_j} = O(\log{n}) \cdot \vc(G)$. 
\end{lemma}
\begin{proof}
	Fix any $j \in [t]$; we prove that $\bO_j \leq 8 \cdot \vc(G)$. The lemma follows from this since there are at most $O(\log{n})$ different sets $\bO_j$ and the union of the sets $O_j$'s is a subset of $\Vvc$ (with size $\vc(G)$). 
	
	Consider the graph $G_j$. The maximum degree in this graph is at most $n/2^{j-1}$ by the definition of the process. Since all the edges in the graph are incident on at least one vertex of $\Vvc$, there can
	be at most $\card{\Vvc} \cdot n/2^{j-1}$ edges between the remaining vertices in $\Vvc$ and $\bVvc$ in $G_j$. Moreover, any vertex in $\bO_j$ has degree at least $n/2^{j+2}$ by definition and hence there can be at 
	most $\paren{\card{\Vvc} \cdot n/2^{j-1}} / \paren{n/2^{j+2}} \leq 8 \card{\Vvc} = 8 \cdot \vc(G)$ vertices in $\bO_j$, proving the claim.  
\end{proof}

We now prove the main relation between the sets $O_j$'s and $\bO_j$'s defined above and the intermediate sets $\Vii_j$'s computed by $\vccs(\Gii)$. 
The following lemma is the heart of the proof. 

\begin{lemma}\label{lem:vc-subset}
	Fix an $i \in [k]$, and let $A_j = \Vii_j \cap \Vvc$ and $B_j = \Vii_j \cap \bVvc$. 
	With probability $1-O(1/n)$, for any $t \in [\Delta]$: 
	\begin{enumerate}
		\item  $\bigcup_{j=1}^{t} A_j \supseteq \bigcup_{j=1}^{t} O_j$. 
		\item $\bigcup_{j=1}^{t} B_j \subseteq \bigcup_{j=1}^{t} \bO_j$.
	\end{enumerate}
\end{lemma}
\begin{proof}
	To simplify the notation, for any $t \in [\Delta]$, we let $A_{< t} = \bigcup_{j=1}^{t-1} A_j$ and $A_{\geq t} = \bigcup_{j=t}^{\Delta} A_j$ (and similarly for $B_j$'s, $O_j$'s, and $\bO_j$'s). We also use $N_{S}(v)$ to denote 
	the neighbor-set of the vertex $v$ in the set $S \subseteq V$. 
	
	Note that the vertex-sets of the graphs $G$ and $\Gii$ are the same and we can ``project'' the sets $O_j$'s and $\bO_j$'s on graph $\Gii$ as well. In other words, we can say a vertex $v$ in $\Gii$ belongs to $O_j$ iff
	$v \in O_j$ in the original graph $G$. In the following claim, we crucially use the fact that the graph $\Gii$ is obtained from $G$ by sampling each edge w.p. $1/k$ to prove that the degree of 
	vertices across different sets $O_j$'s (and $\bO_j$'s) in $\Gii$ are essentially the same as in $G$ (up to the scaling factor of $1/k$). 
	
		\begin{claim}\label{clm:vc-degree-concentration}
		For any $j \in [\Delta]$: 
		\begin{itemize}
		\item  For any vertex $v \in O_j$, $\card{N_{\bO_{\geq j}}(v)} \geq n/(k \cdot 2^{j+1})$ in the graph $\Gii$ w.p. $1-O(1/n^2)$. 
		\item For any vertex $v \in \bO_{\geq j+1}$, $\card{N_{O_{\geq j}}(v)} < n/(k \cdot 2^{j+1})$ in the graph $\Gii$ w.p. $1-O(1/n^2)$. 
		\end{itemize}
	\end{claim}
	\begin{proof}
		Fix any $j \in [\Delta]$ and $v \in O_j$. By definition of $O_j$, degree of $v$ is at least $n/2^{j}$ in $G_j$; in other words, $\card{N_{\bO^{\geq j}}(v)} \geq n/2^{j}$ in the graph $G$. Since each edge in $G$ is 
		sampled w.p. $1/k$ in $\Gii$, $\card{N_{\bO^{\geq j}}(v)} \geq n/(k\cdot 2^{j})$ in $\Gii$ in expectation. Moreover, by the choice of $\Delta$, $n/(k\cdot 2^{j}) \geq 4\log{n}$, and hence by Chernoff bound, 
		w.p. $1-O(1/n^2)$, $\card{N_{\bO^{\geq j}}(v)} \geq n/(k\cdot 2^{j+1})$ in $\Gii$.
		
		Similarly for a vertex $v \in \bO^{\geq j+1}$, degree of $v$ is less than $n/2^{j+2}$ in $G_j$ by definition of $\bO_j$; hence, $\card{N_{O^{\geq j}}(v)} < n/2^{j+2}$ in the graph $G$. Using a similar argument as before, 
		by Chernoff bound, w.p. $1-O(1/n^2)$, $\card{N_{O^{\geq j}}(v)} < n/(k\cdot 2^{j+1})$ in $\Gii$.
	\end{proof}
	
	By using a union bound on the $n$ vertices in $G$, the statements in Claim~\ref{clm:vc-degree-concentration} hold simultaneously for all vertices of $G$ w.p. $1-O(1/n)$; in the following we 
	condition on this event. We now prove Lemma~\ref{lem:vc-subset} by induction. 
	
	Let $v$ be a vertex that belongs to $O_1$; we prove that $v$ belongs to the set $\Vii_1$ of $\vccs$, i.e., $v \in A_1$. By Claim~\ref{clm:vc-degree-concentration} (for $j=1$), the degree of $v$ in $\Gii_1$ is 
	at least $n/4k$. Note that in $\Gii_1$, $v$ may also have edges to other vertices in $\Vvc$ but this can only increase the degree of $v$. This implies that $v$ also belongs to $A_1$ by the threshold chosen
	in $\vccs$. Similarly, let $u$ be a vertex in $\bO_{\geq 2}$ (i.e., \emph{not} in $\bO_1$); we show that $u$ is not chosen in $\Vii_1$, implying that $B_1$ can only contain vertices in $\bO_1$. 
	By Claim~\ref{clm:vc-degree-concentration}, degree of $u$ in $\Gii_1$ is less than $n/4k$. This implies that $u$ does not belong 
	to $B_1$. In summary, we have $O_1 \subseteq A_1$ and $B_1 \subseteq \bO_1$. 
	
	Now consider some $t > 1$ and let $v$ be a vertex in $O_t$. By induction, $B_{< t} \subseteq \bO_{<t}$. This implies that the degree of $v$ to $B_{\geq t}$ is at least as large as its degree to $O_{\geq t}$. Consequently,
	by Claim~\ref{clm:vc-degree-concentration} (for $j = t$), degree of $v$ in the graph $\Gii_t$ is at least $n/(k \cdot 2^{t+1})$ and hence $v$ also belongs to $A_{t}$. Similarly, fix a vertex $u$ in $\bO_{\geq t+1}$. By induction, 
	$A_{<t} \supseteq O_{<t}$ and hence the degree of $u$ to $A_{\geq t}$ is at most as large as its degree to $O_{\geq t}$; note that since $\Vvc$ is a vertex cover, $u$ does not have any other edge
	in $\Gii_t$ except for the ones to $A_{\geq t}$. We can now argue as before that $u$ does not belong to $B_{t}$. 
\end{proof}

We are now ready to prove Theorem~\ref{thm:vc}. 

\begin{proof}[Proof of Theorem~\ref{thm:vc}]
	The bound on the coreset size follows immediately from the fact that the graph $\Gii_{\Delta}$ contains at most $O(n\log{n})$ edges and size of $\Vcsii$ is at most $n$. 
	As argued before, to prove the bound on the approximation ratio, we only need to show that $\bigcup_{i=1}^{k} \Vcsii$ is of size $O(\log{n}) \cdot \vc(G)$.
	Let $A^{(i)} = \Vcsii \cap \Vvc$ and $B^{(i)} = \Vcsii \cap \bVvc$; clearly, each $A^{(i)} \subseteq \Vvc$ and moreover, by Lemma~\ref{lem:vc-subset} (for $t = \Delta$), each $B^{(i)} \subseteq \cup_{j=1}^{\Delta} \bO_j$.  
	Consequently, $\card{\bigcup_{i=1}^{k} \Vcsii}\leq \card{\Vvc} + \card{\bigcup_{j=1}^{\Delta} \bO_j} \leq O(\log{n}) \cdot \vc(G)$, where the last inequality is by Lemma~\ref{lem:vc-analysis-ratio}. 
\end{proof}


\section{Lower Bounds for Randomized Coresets}\label{sec:lower-coreset}
We formalize Result~\ref{res:lb-coreset} in this section. As argued earlier, Result~\ref{res:lb-coreset} is a special case of Result~\ref{res:lb-dist} and hence follows from that result; however, as the proof of 
Result~\ref{res:lb-dist} is rather technical and complicated, we also provide a self-contained proof of Result~\ref{res:lb-coreset} as a warm-up to Result~\ref{res:lb-dist}.  

\subsection{A Lower Bound for Randomized Composable Coresets of Matching} \label{sec:lb-matching-coreset}

We prove a lower bound on the size of any randomized composable coreset for the matching problem, formalizing Result~\ref{res:lb-coreset} for matching. 

\begin{theorem}\label{thm:matching-lb-coreset}
	For any $k = o(n/\log{n})$ and $\alpha = o(\min\set{n/k,k})$, any $\alpha$-approximation randomized composable coreset of the maximum matching problem is of size $\Omega(n/\alpha^2)$. 
\end{theorem}

By Yao's minimax principle~\cite{Yao87}, to prove the lower bound in Theorem~\ref{thm:matching-lb-coreset}, it suffices to analyze the 
performance of deterministic algorithms over a fixed (hard) distribution. We propose the following distribution for this task. 
For simplicity of exposition, in the following, we prove a lower bound
for $(\alpha/4)$-approximation algorithms; a straightforward scaling of the parameters proves the lower bound for $\alpha$-approximation. 

\textbox{Distribution $\distMatch$. \textnormal{A hard input distribution for the matching problem.}}{
\begin{itemize}
	\item Let $G(L,R,E)$ (with $\card{L} = \card{R} = n$) be constructed as follows: 
		\begin{enumerate}
			\item Pick $A \subseteq L$ and $B \subseteq R$, each of size $n/\alpha$, uniformly at random. 
			\item Define $\Eab$ as a set of edges between $A$ and $B$, chosen by picking each edge in $A \times B$ w.p. $k \cdot \alpha/n$. 
			\item Define $\Ebab$ as a \emph{random} perfect matching between $\bA$ and $\bB$. 
			\item Let $E:= \Eab \cup \Ebab$. 
		\end{enumerate}
	\item Let $\Ei{1},\ldots,\Ei{k}$ be a \emph{random $k$-partitioning} of $E$ and let the input to player $\Ps{i}$ be the graph $\Gii(L,R,\Eii)$. 
\end{itemize}
} 

Let $G$ be a graph sampled from the distribution $\distMatch$. Notice first that the graph $G$ always has a matching of size at least $n-n/\alpha \geq n/2$, i.e., the matching $\Ebab$. 
Additionally, it is easy to see that any matching of size more than $2n/\alpha$ in $G$ uses at least $n/\alpha$ edges from $\Ebab$: the edges in $\Eab$ can only form a matching of size $n/\alpha$ by construction. 
This implies that any $(\alpha/4)$-approximate solution requires recovering at least $n/\alpha$ edges from $\Ebab$. In the following, we prove that this is only possible if the coresets of the players are sufficiently large. 


For any $i \in [k]$, define the \emph{induced matching} $\Mii$ as the unique matching in $\Gii$ that is incident on \emph{vertices of degree exactly one}, i.e., both end-points of each edge
in $\Mii$ have degree one in $\Gii$. We emphasize that the notion of induced matching is with respect to the entire graph and not only with respect to the vertices included in the induced matching.
We have the following crucial lemma on the size of $\Mii$. 
The proof is technical and is deferred to Appendix~\ref{app:large-induced-matching}.
\begin{lemma}\label{lem:matching-dist-large}
	W.p. $1-O(1/n)$, for all $i \in [k]$, $\card{\Mii} = \Theta(n/\alpha)$. 
\end{lemma}

We are now ready to prove Theorem~\ref{thm:matching-lb-coreset}.

\begin{proof}[Proof of Theorem~\ref{thm:matching-lb-coreset}]

Fix any randomized composable coreset (algorithm) for the matching problem that has size $o(n/\alpha^2)$. We show that such a coreset cannot achieve a better than $(\alpha/4)$-approximation
over the distribution $\distMatch$. As argued earlier, to prove this, we need to show that this coreset only contains $o(n/\alpha)$ edges from $\Ebab$ in expectation. 

Fix any player $i \in [k]$, and let $\Mstari$ be the subset of the matching $\Ebab$ assigned to $\Ps{i}$. It is clear that $\Mstari \subseteq \Mii$ by the definition of $\Mii$. 
Moreover, define $X_i$ as the random variable denoting the number of edges from $\Mstari$ that belong to the coreset sent by player $\Ps{i}$. Notice that $X_i$ is clearly an upper bound
on the number of edges of $\Ebab$ that are in the final matching of coordinator and also belong to the input graph of player $\Ps{i}$. In the following, we show that 
\begin{align}
	\Ex\bracket{X_i} = o\paren{\frac{n}{k \cdot \alpha}} \label{eq:show-lb-coreset}
\end{align}
Having proved this, we have that the expected size of the output matching by the coordinator is
at most $n/\alpha + \sum_{i=1}^{k} \Ex\bracket{X_i} = n/\alpha + o(n/\alpha) < (\alpha/4) \cdot \mm(G)$, a contradiction. 

We now prove Eq~(\ref{eq:show-lb-coreset}). 
In the following, we condition on the event that $\card{\Mstari} = \Theta(n/k)$ and $\card{\Mii} = \Theta(n/\alpha)$; by Chernoff bound (for the first part, since $n/k = \omega(\log{n})$) 
and Lemma~\ref{lem:matching-dist-large} (for the second part), 
this event happens with probability $1-O(1/n)$. As such, this conditioning can only change $\Ex\bracket{X_i}$ by an additive factor of $O(1)$ which we ignore in the following.

A crucial property of the distribution $\distMatch$ is that the edges in $\Mstari$ and the remaining edges in $\Mii$ are {indistinguishable} in $\Gii$. 
More formally, for any edge $e \in \Gii$, 
\begin{align*}
	\Pr\paren{e \in \Mstari \mid e \in \Mii} = \frac{\card{\Mstari}}{\card{\Mii}} = \Theta(\alpha/k)
\end{align*}

On the other hand, for a fixed input $\Mii$ to player \Ps{i}, the computed coreset $C_i$ is always the same (as the coreset is a deterministic function of the player input). Hence, 
\begin{align*}
	\Ex\bracket{X_i} = \sum_{e \in C_i} \Pr\paren{e \in \Mstar_i \mid e \in \Mii} = \card{C_i} \cdot \Theta(\alpha/k) = o(n/\alpha^2) \cdot \Theta(\alpha/k) = o\paren{n/(\alpha \cdot k)}
\end{align*}
where the second last equality is by the assumption that the size of the coreset, i.e., $\card{C_i}$, is $o(n/\alpha^2)$. This finalizes the proof. 
\end{proof}

\subsection{A Lower Bound for Randomized Composable Coresets of Vertex Cover}\label{sec:lb-vc-coreset}

In this section, we prove that the size of the corset for the vertex cover problem in Theorem~\ref{thm:vc} is indeed optimal. The following is a formal statement of Result~\ref{res:lb-coreset} for the vertex cover problem.   
\begin{theorem}\label{thm:vc-lb-coreset}
	For any $k = o(n/\log{n})$ and $\alpha = o(\min\set{n/k,k})$, any $\alpha$-approximation randomized composable coreset of the minimum vertex cover problem is of size $\Omega(n/\alpha)$. 
\end{theorem}

By Yao's minimax principle~\cite{Yao87}, to prove the lower bound in Theorem~\ref{thm:vc-lb-coreset}, it suffices to analyze the 
performance of deterministic algorithms over a fixed (hard) distribution. We propose the following distribution for this task\footnote{We point out that simpler versions of this distribution suffice for proving the lower bound
in this section. However, as we would like this proof to also act as a warm-up to the proof of Theorem~\ref{thm:vc-lb}, we use the same distribution that is used to prove that theorem.}.
For simplicity of exposition, in the following, we prove a lower bound for $(c \cdot \alpha)$-approximation algorithms (for some constant $c > 0$); 
a straightforward scaling of the parameters proves the lower bound for $\alpha$-approximation as well. 

\textbox{Distribution $\distVC$. \textnormal{A hard input distribution for the vertex cover problem.}}{
\begin{itemize}
	\item Construct $G(L,R,E)$ (with $\card{L} = \card{R} = n$) as follows: 
		\begin{enumerate}
			\item Pick $A \subseteq L$ of size $n/\alpha$ uniformly at random. 
			\item Let $\EA$ be a set of edges chosen by picking each edge in $A \times R$ w.p. $k/2n$. 
			\item Pick a single vertex $\vstar$ uniformly at random from $\bA$ and let $\estar$ be an edge incident on $\vstar$ chosen uniformly at random. 
			\item Let $E:= \EA \cup \set{\estar}$. 
		\end{enumerate}
	\item \label{line:r-partitioning} Let $\Ei{1},\ldots,\Ei{k}$ be a \emph{random $k$-partitioning} of $E$ and let the input to player $\Ps{i}$ be the graph $\Gii(L,R,\Eii)$. 
\end{itemize}
}

For any $i \in [k]$, we define $L^1_i$ as the set of vertices in $L$ with degree \emph{exactly one} in $G^{(i)}$. We further define $R^1_i$ as the 
set of neighbors of vertices in $L^1_i$ (note that vertices in $R^1_i$ do not  \emph{not} necessarily have degree exactly one). 
We start by proving a simple property of this distribution. 

\begin{lemma}\label{lem:vc-lb-degree-one}
	For any $i \in [k]$, $\card{L^1_i} = \Theta(n/\alpha)$ and $\card{R^1_i} = \Theta(n/\alpha)$ w.p. $1-o(1)$. 
\end{lemma}
\begin{proof}
	Fix any player $i \in [k]$ and any vertex $v \in A$. The distribution of neighborhood of $v$ in the graph $\Gii$ is as follows: pick each vertex in $R$ w.p. $1/2n$ independently;
	this is because each vertex in $R$ is chosen w.p. $k/2n$ to be a neighbor of $v$ in $G$ and then each of these vertices are assigned to the graph $\Gii$ w.p. $1/k$ by the random $k$-partitioning. 
	As such, 
	\begin{align*}
		\Pr\paren{d(v) = 1~\text{in $G^{(i)}$}} = {{n} \choose {1}} \cdot \frac{1}{2n} \cdot \paren{1-\frac{1}{2n}}^{n-1} \approx \frac{1}{2\sqrt{e}} =  \Theta(1)
	\end{align*}	
	Consequently, we have $\Ex\bracket{\card{L^1_i}} = \card{A} \cdot \Theta(1) = \Theta(n/\alpha)$ 
	and by Chernoff bound, $\card{L^1_i} = \Theta(n/\alpha)$ (note that for one player $\vstar$ would also belong to $O_i$ but that only changes
	the size of $\card{O_i}$ by one vertex). 
	
	We now bound the size of $R^1_i$. Each vertex in $L^1_i$ is choosing one vertex uniformly at random from $R$ and hence we can model this distribution by a simple balls and bins experiment (throwing $\card{L^1_i}$ balls
	into $n$ bins, each independently and uniformly at random), and hence by a standard fact about balls and bins experiments argue that $\card{R^1_i} = \Theta(n/\alpha)$ w.p. $1-o(1)$ as well (see
	Proposition~\ref{prop:random1} in Appendix~\ref{app:large-induced-matching} for a proof of this fact about balls and bins experiments). 
\end{proof}

We can now prove Theorem~\ref{thm:vc-lb-coreset}. 

\begin{proof}[Proof of Theorem~\ref{thm:vc-lb-coreset}]
	Let $i$ be the index of the player \Ps{i} that the edge $\estar$ is given to. We argue that if the coreset sent by player $\Ps{i}$ is of size $o(n/\alpha)$, then the coordinator cannot obtain a vertex
	cover of size $o(n)$. As the graph $G$ admits a vertex cover of size $(n/\alpha+1)$ (pick $A$ and $\vstar$), this proves the theorem. 
	
	By Lemma~\ref{lem:vc-lb-degree-one}, the set of vertices in $L$ with degree exactly one in $\Gii$ and the set of their neighbors in $R$, i.e., the sets $L^1_i$ and $R^1_i$, are of size $\Theta(n/\alpha)$ w.p. $1-o(1)$. 
	In the following, we condition on this event. 
	As the algorithm used by $\Ps{i}$ to create the coreset is deterministic, given a fixed input, it always creates the same coreset. However, a crucial property of the distribution $\distVC$ is that, conditioned on a fixed assignment to 
	$L^1_i$, the vertex $\vstar$ is chosen uniformly at random from $L^1_i$. This implies that if the coreset of player $\Ps{i}$ contains $o(n/\alpha)$ edges, then w.p. $1-o(1)$, $\estar$ is not part of the coreset ($\estar$ is
	chosen uniformly at random from the set of all edges incident on $L^1_i$). 
	Similarly, if the coreset fixes $o(n/\alpha)$ vertices to be added to the final solution, w.p. $1-o(1)$, no end point of $\estar$ is added to this 
	fixed set ($\vstar$ is chosen uniformly at random from $L^1_i$ of size $\Theta(n/\alpha)$, and the other end point of $\estar$ is chosen uniformly at random from $R^1_i$ of size $\Theta(n/\alpha)$). 
	Finally, the coresets of other players are all independent of the edge $\estar$ and hence as long as the total number of fixed vertices sent by the players is $o(n)$, w.p. $1-o(1)$, no end points 
	of $\estar$ are present in the fixed solution. Conditioned on these three events, w.p. $1-o(1)$, the output of the algorithm does not cover the edge $\estar$ and hence is not a feasible vertex cover. 
	
	We remark that this argument holds even if we are allowed to add extra vertices to the final vertex cover (other than the ones fixed by the players or computed
	 as a vertex cover of the edges in the coresets), since conditioned on $\estar$ not being present in any coreset,
	the end point of this edge are chosen uniformly at random from all vertices in $L \setminus A$ and $R$ and hence a solution of size $o(n)$ would not contain either of them w.p. $1-o(1)$. 
\end{proof}

\section{Communication Complexity Lower Bounds}\label{sec:lower-dist}

We prove Result~\ref{res:lb-dist} in this section, showing that our randomized composable coresets in fact obtain the optimal communication complexity (among all possible protocols, not necessarily a coreset)
in the simultaneous communication model. This result is a vast generalization of Result~\ref{res:lb-coreset} proved in Section~\ref{sec:lower-coreset}. 

\subsection{An $\Omega(nk/\alpha^2)$ Lower Bound on Communication Complexity of Matching} \label{sec:matching-dist}

We prove a lower bound on the simultaneous communication complexity of the matching problem in the random partition
model, formalizing Result~\ref{res:lb-dist} for matching. 

\begin{theorem}\label{thm:matching-lb}
	For any $\alpha$ between $\Omega(\log{n})$ and $o(\min\set{\sqrt{n/k},k})$, the simultaneous communication complexity of $\alpha$-approximating the matching problem in the random partition
	model is $\Omega(nk/\alpha^2)$. 
\end{theorem}

By Yao's minimax principle~\cite{Yao87}, it suffices to analyze the communication complexity of 
\emph{deterministic} protocols over a fixed (hard) distribution. We again use the distribution $\distMatch$ in Section~\ref{sec:lb-matching-coreset}.
For the convenience of the reader, we repeat the description of this distribution here. 

\textbox{Distribution $\distMatch$. \textnormal{A hard input distribution for the matching problem.}}{
\begin{itemize}
	\item Let $G(L,R,E)$ (with $\card{L} = \card{R} = n$) be constructed as follows: 
		\begin{enumerate}
			\item Pick $A \subseteq L$ and $B \subseteq R$, each of size $n/\alpha$, uniformly at random. 
			\item Define $\Eab$ as a set of edges between $A$ and $B$, chosen by picking each edge in $A \times B$ w.p. $k \cdot \alpha/n$. 
			\item Define $\Ebab$ as a \emph{random} perfect matching between $\bA$ and $\bB$. 
			\item Let $E:= \Eab \cup \Ebab$. 
		\end{enumerate}
	\item Let $\Ei{1},\ldots,\Ei{k}$ be a \emph{random $k$-partitioning} of $E$ and let the input to player $\Ps{i}$ be the graph $\Gii(L,R,\Eii)$. 
\end{itemize}
} 

Let $G$ be a graph sampled from the distribution $\distMatch$. Notice first that the graph $G$ always has a matching of size at least $n-n/\alpha \geq n/2$, i.e., the matching $\Ebab$. 
Additionally, it is easy to see that any matching of size more than $2n/\alpha$ in $G$ uses at least $n/\alpha$ edges from $\Ebab$: the edges in $\Eab$ can only form a matching of size $n/\alpha$ by construction. 
This implies that any $(\alpha/4)$-approximate solution requires recovering at least $n/\alpha$ edges from $\Ebab$\footnote{Similar to Section~\ref{sec:lb-matching-coreset}, we prove the lower bound
for $(\alpha/4)$-approximation protocols for simplicity of representation.}. It is this task that we show requires $\Omega(nk/\alpha^2)$ communication. 

The following definitions are identical to those in Section~\ref{sec:lb-matching-coreset}.  
For any $i \in [k]$, define the \emph{induced matching} $\Mii$ as the unique matching in $\Gii$ that is incident on \emph{vertices of degree exactly one}, i.e., both end-points of each edge
in $\Mii$ have degree one in $\Gii$. Recall that by Lemma~\ref{lem:matching-dist-large}, $\card{\Mii} = \Theta(n/\alpha)$ w.h.p.

Let $\Mstari$ be the subset of the matching $\Ebab$ assigned to $\Ps{i}$. It is clear that $\Mstari \subseteq \Mii$ by the definition of $\Mii$. By our previous discussion, it is these edges of $\Mstari$ that the player $\Ps{i}$ needs to 
communicate to the coordinator. Moreover, notice that the players can simply ignore all edges in $\Gii$ that do not belong to $\Mii$ as they clearly cannot be in $\Mstari$. 
However, a crucial property of the distribution $\distMatch$ is that the edges in $\Mstari$ and the remaining edges in $\Mii$ are \emph{indistinguishable} in $\Gii$. In other words, conditioned on a 
specific assignment for $\Mii$, \emph{any} edge $e \in \Mii$ belongs to the matching $\Mstari$ w.p. $\alpha/k$ independent of the other edges. Moreover, it is intuitive to think that only player 
$\Ps{i}$ is able to communicate the edges in $\Mstari$ as the input of other players, while dependent on the set of vertices in the matching $\Mii$, are essentially independent 
of the edges in $\Mii$. This discussion suggests the following intermediate problem in the one-way two-player communication model. 

\begin{problem}[\MatchingRecovery problem]\label{prob:matching-recovery}
	Let $H$ be a bipartite graph with $t$ vertices on each side. Alice is given a perfect matching $\MA$ in $H$ and Bob is given the following input: 
	\begin{itemize}
		\item Two sets $P$, $Q$ (each in one side of the bipartition of $H$) with $\card{P} = \card{Q} = p$ and the \emph{promise} that in matching $\MA$ vertices in $P$ are matched to vertices $Q$. 
		\item A set $\EB$ of edges in $H$ with the \emph{promise} that the matching $\MA$ does not contain any edge from $\EB$.  
	\end{itemize}
	The goal is for Alice to send a message to Bob and Bob needs to output the edges in the matching $\MA$ that are between $P$ and $Q$. 
\end{problem}

Consider the following distribution $\distMD$ for $\MatchingRecovery$ based on the distribution $\distMatch$: Fix any arbitrary $i \in [k]$; we sample an input instance $\Gi{1},\ldots,\Gi{k}$ from $\distMatch$. Then, 
we let $H$ be the bipartite graph on the set of vertices in $\Mii$ and let $\MA = \Mii$. We let the input sets $P$ and $Q$ to Bob be the set of vertices incident on $\Mstari$. Finally, we let $\EB$ be the set
of edges assigned to all graphs $\Gi{j}$ for $j \neq i$ that are between the vertices matched by $\Mii$, i.e., are inside the graph $H$. This completes the description of the distribution $\distMD$. 

In the following, we condition on the inputs chosen from $\distMD$ to have the following additional property: $\card{\Mii} = \Theta(n/\alpha)$ and $\card{\Mstari} = \Theta(n/k)$. 
Notice that by Lemma~\ref{lem:matching-dist-large}, w.p. $1-O(1/n)$, for any player \Ps{i}, $\card{\Mii} = \Theta(n/\alpha)$. A simple application of Chernoff bound also ensures that the
number of edges from $\Ebab$ assigned to each player, i.e., the edges in the matching $\Mstari$ is $\Theta(n/k)$ w.p. $1-O(1/n)$. Consequently, conditioning on this event is essentially not changing the 
distribution and hence for simplicity from now on, we always assume the inputs chosen from $\distMD$ satisfy the mentioned properties. We establish the following lower bound for \MatchingRecovery. 

\begin{lemma}[Communication Complexity of \MatchingRecovery]\label{lem:cc-matching-detection}
	Suppose $s = \Omega(k)$ denotes the communication cost of a protocol for \MatchingRecovery and $X$ denotes the number of edges output by Bob in this protocol; for the inputs chosen
	from the distribution $\distMD$, we have $\Ex\bracket{X} \leq (\alpha/k) \cdot O(s)$. 
\end{lemma}

Notice that the distribution $\distMD$ for \MatchingRecovery imposes a non-trivial correlation between the inputs of the two players which complicates the proof of this lower bound.
We address this issue by expressing this distribution as a convex combination of a relatively small number of simpler (yet non-trivial) distributions and prove the lower bound for each distribution separately. 
These distributions are still not product distributions but we can show that the mild correlation in the input of the players in this case can be managed directly using a careful combinatorial argument. 
The proof is deferred to Section~\ref{sec:matching-detection}. Before that, we prove a formal reduction from the matching problem to \MatchingRecovery and use 
Lemma~\ref{lem:cc-matching-detection} to finalize the proof of Theorem~\ref{thm:matching-lb}. 

\begin{proof}[Proof of Theorem~\ref{thm:matching-lb}] 
	Fix a protocol $\ProtMatching$ for the matching problem on $\distMatch$ with communication cost $o(nk/\alpha^2)$ 	and
	suppose each player $\Ps{i}$ communicates at most $s_i$ bits in this protocol. We assume that $s_i = \Omega(k)$ as otherwise we can simply augment it with $\Omega(k)$ bits to satisfy this bound
	while increasing the total communication cost of the protocol by $O(k^2) = o(nk/\alpha^2)$ bits (since $\alpha = o(\sqrt{n/k})$). Our goal is to show that in this protocol, at most $o(n/\alpha)$ edges
	from $\Ebab$ can be matched in expectation. The result then follows from the fact that obtaining better than $(\alpha/4)$-approximation requires outputting $\Omega(n/\alpha)$ edges
	from $\Ebab$. 

	We use $\ProtMatching$ to create $k$ protocols for \MatchingRecovery, whereby in the $i$-th protocol $\Prot_i$, 
	Alice plays the role of player $\Ps{i}$ and Bob plays the role of all other players plus the coordinator. 
	Fix an $i \in [k]$; the protocol $\Prot_i$ works as follows. 
	
	Given their input in \MatchingRecovery, Alice and Bob sample two random sets $\XM \subseteq L$ and $\YM \subseteq R$, each of size $t$ using public randomness. They also 
	sample two random sets $\XbM \subseteq L \setminus \XM$ and $\YbM \subseteq R \setminus \YM$, each of size $(n/\alpha-t+p)$. \
	
	Alice creates the graph $\Gii$ by letting $\Mii$ be 
	the matching $\MA$ and choose the remainder of the graph $\Gii$ by sampling the edges between $\XbM$ and $\YbM$ using the same distribution as the distribution of $\Gii$ conditioned 
	on $\Mii = \MA$ and the set of non-zero degree vertices in $\Gii$ being subset of $\XbM$ and $\YbM$. 
	
	Bob creates the input of the other players as follows. Bob picks a random mapping $\sigma: \EB \rightarrow [k] \setminus \set{i}$ and 
	in the graph $\Gi{j}$ for $j \neq i$, he assigns $\sigma^{-1}(j) \subseteq \EB$ to be the edges between $\XM$ and $\YM$. Finally, Bob samples
	the remainder of the graphs from the joint distribution of $\distMatch$ conditioned on the set $A = \XbM \cup \XM \setminus P$, $B = \XbM \cup \XM \setminus Q$,
	and the edges between $\XM$ and $\YM$ be the ones already sampled (via $\sigma$). Note that since Bob has the knowledge of the sets $P$ and $Q$, he can sample the reminder of the matching 
	$\Ebab$ for the $k-1$ remaining players as well. 
	
	One can verify that the distribution of the instances created by this reduction matches the distribution $\distMatch$. To finish
	the reduction, Alice and Bob simulate the protocol $\ProtMatching$ by Alice sending the message of $\Ps{i}$ to Bob (or equivalently the coordinator) and 
	Bob creating the message of all other players locally and completing the protocol. Bob then outputs the part of the matching computed by $\ProtMatching$ which lies between $P$ and $Q$. 
	This results in a protocol $\Prot_i$ for $\MatchingRecovery$ with communication cost of $s_i$ bits.

	We can now invoke Lemma~\ref{lem:cc-matching-detection} to argue that the expected number of edges
	matched by $\Prot_i$ and hence by $\ProtMatching$ for the player $\Ps{i}$ in the distribution $\distMatch$ is at most $(\alpha/k) \cdot O(s_i)$. Summing over all players, 
	we have that the total number of matched edges in $\Ebab = \bigcup_{i=1}^{k} \Mstari$ is 
	$ \sum_{i=1}^{k} (\alpha/k) \cdot O(s_i) = (\alpha/k) \cdot o(nk/\alpha^2) = o(n/\alpha)$, 
	where the first equality is by the bound on the communication cost of the protocol. This completes the proof. 
\end{proof}

We remark that the bound in Theorem~\ref{thm:matching-lb} is tight (up to an $O(\log{n})$ factor) for all ranges of $\alpha$. 

\begin{remark}\label{rem:tight-matching}
	The protocol in which each player computes a maximum matching of the input graph, subsamples the edges of this matching w.p. $1/\alpha$, and sends it to the coordinator who outputs a maximum matching
	of the received matchings is an $\alpha$-approximation protocol for the maximum matching problem with $\Ot(nk/\alpha^2)$ total communication. 
\end{remark}

We briefly sketch the proof of correctness for protocol in Remark~\ref{rem:tight-matching}. Assume first that the maximum matching size in each player input is of size $\Ot(n/\alpha)$. The bound on the 
total communication cost follows immediately from this assumption. To see the correctness, recall that in the proof of Theorem~\ref{thm:matching}, we showed each coreset (here, each player) can increase the size 
of the output matching by $\Omega(\mm(G)/k)$; since we are subsampling the maximum matching by a factor of $\alpha$, this increment would be $\Omega(\mm(G)/\alpha k)$ and hence over all $k$ players, 
we obtain a matching of size $\Omega(\mm(G)/\alpha)$. The assumption on the size of the maximum matching in each player input is essentially without loss of generality since otherwise one player can send an 
$\alpha$-approximate matching to the coordinator alone, resulting in a protocol with $\Ot(n/\alpha)$ communication. We point out that a simple concentration result proves that the 
maximum matching size between players is concentrated within an $O(\log{n})$ factor. This easily implies that in this case we can ensure that only one player is sending his maximum matching and not all players.

\subsection{Communication Complexity of \MatchingRecovery}\label{sec:matching-detection}

 We start by reformulating the distribution $\distMD$ to make it more suitable for proving the lower bound in Lemma~\ref{lem:cc-matching-detection}. 
Indeed, the distribution $\distMD$ is \emph{not} a product distribution: the promise that Alice's matching $\MA$ needs to always match the set $P$ to $Q$ correlates Alice's and Bob's input in a non-trivial way, 
complicating the analysis. To address this, we show that the distribution $\distMD$ can be expressed as a convex combination of a relatively small set of (essentially) product distributions; this significantly simplifies the proof of 
the lower bound. 

Let us first define the distribution $\distMD$ directly, i.e., without depending on the distribution $\distMatch$. 
In $\distMD$ conditioned on $\card{\MA} = t$ and $\card{P} = \card{Q} = p$, the matching $\MA$ is chosen uniformly at
random from the set of all matchings of size $t$ in $H$ and then $P$ and $Q$ are chosen uniformly at random from all pairs of sets of size $p$ that
are matched together in $\MA$. Finally, the edge-set $\EB$ is chosen by picking each edge in $H$ that is not incident on $P$ and $Q$ and not in $\MA$ w.p. $(k-1) \cdot \alpha/n$. One can check
that this results in an equivalent definition of $\distMD$.  

We now reformulate the distribution $\distMD$ as follows: we first randomly partition the vertices in $L_H$ and $R_H$ (i.e., the bipartition of the graph $H$) into $c:= \floor{t/p}$ \emph{blocks} denoted
by $\vecB:= (P_1,Q_1),\ldots,(P_c,Q_c)$ such that $P_i \subseteq L_H$, $Q_i \subseteq R_H$ and $\card{P_i} = \card{Q_i} = p$ for all $i \in [c]$. We then create $\MA$ by picking a random matching that 
matches each $P_i$ to $Q_i$\footnote{Note that this way, it is possible that up to $p$ vertices in $L_H$ and $R_H$ become ``left overs'' i.e., do not belong to any block. We pick a random matching between these vertices
also to complete the description of $\MA$.}. Finally, the input to Bob is a pair $P$ and $Q$ chosen uniformly at random from these $c$ blocks. We pick the edge-set $\EB$ of Bob as before. It is again easy to
verify that this is indeed an equivalent formulation of the distribution $\distMD$. 

Suppose we provide the identity of the blocks $\vecB$ to both Alice and Bob; this essentially breaks the dependence between Alice's and Bob's inputs (except for the mild correlation enforced by $E_B$ that 
we deal with directly). Note that revealing this extra information can only make our lower bound result stronger. In the following, we argue that even with $\vecB$ being public information, to solve the problem, Alice needs to communicate a large message. 

For any fixed $\vecB$, let $\sigma_{\vecB}(\MA)$ be the (deterministic) mapping used by Alice to create the message sent to Bob. The mapping $\sigma_{\vecB}$ then partitions 
all matchings $\MA$ that are valid with respect to $\vecB$ into $2^{s}$ classes $\Sigma_1,\ldots,\Sigma_{2^s}$ (one per each message). 
Moreover, for $\vecB$ and a fixed set of edges $\EB$, we define $\FC(\EB,\vecB)$ as the set of matchings $\MA$ that conform to the restrictions
imposed by both $\EB$ and $\vecB$ (we use $\FC$ if $\EB$ and $\vecB$ are clear from the context). Note that $\sigma_{\vecB}$ similarly partitions $\FC$ into $2^{s}$ different classes as well.

An important observation is that given a message corresponding to some class $\Sigma_i$ and the edges $\EB$, the matching $\MA$ is chosen uniformly at random
from all matchings in $\Sigma_i \cap \FC$; consequently, Bob can only output an edge in the final answer if it belongs to 
\emph{all} matchings in $\FC$ that are mapped to $\Sigma_i$. For any set $F \subseteq \FC$, we define $M_F$ as the intersection of all matchings $\MA$ in $F$. Intuitively, whenever $M_F$ is large, the set
$F$ itself should be small as many edges of the matchings $\MA \in F$ are forced to be the same. We formalize this intuition in the following lemma. 

\begin{lemma}\label{lem:matching-counting}
	For any set $F \subseteq \FC(\EB,\vecB)$, if $M_F$ contains $\ell$ edges, then $\card{F} \leq {2^{-\paren{\ell - \Theta(k)}}} \cdot \card{\FC}$ w.h.p. (the probability is taken over the choice of $E_B$ after fixing $\vecB$). 
\end{lemma}
\begin{proof}
	Note that we can switch the order in which we pick $\MA$ and $\EB$ in the distribution $\distMD$. Fix any block $(P_i,Q_i) \in \vecB$ and any vertex $u \in P_i$; 
	each edge $(u,v)$ for $v \in Q_i$ is chosen w.p. $(k-1) \cdot \alpha/n$ in $\EB$. Hence, w.h.p, at most $\beta:= 2\cdot \card{Q_i} \cdot (k-1) \cdot \alpha/n = \Theta(\alpha)$
	neighbors are chosen for $u$ in $Q_i$ in $\EB$ (here, we used a standard application of Chernoff bound and the
	assumption that $\alpha = \Omega(\log{n})$). This means that for each vertex in $P_i$ there are $p - \beta$ possible choices for its neighbor in $\MA$; hence, there
	are at least $(p-\beta)!$ choices for $\MA$ to match $P_i$ to $Q_i$. Since there are $c$ different blocks, we have $\card{\FC} \geq \paren{(p-\beta)!}^{c}$.
	
	Now suppose we fix $\ell$ edges for the matching $\MA$ (as happens in the set $F$), and let $\ell_i$ be the number of fixed edges between $(P_i,Q_i) \in \vecB$. There can be at most $(p-\ell_i)!$ choices 
	for the matching between $P_i$ and $Q_i$ (we ignore the restriction implied by $\EB$ for the purpose of obtaining an upper bound). Hence, 
	\begin{align*}
		\frac{\card{\FC}}{\card{F}} &\geq \prod_{i=1}^{c} \frac{(p-\beta)!}{(p-\ell_i)!} \geq \prod_{i=1}^{c} (p-\beta)  \ldots (p-\beta-\ell_i+1) \\
		&\geq \prod_{i=1}^{c} 2^{\ell_i - \beta} = 2^{\ell - c \cdot \beta} = 2^{\ell - \Theta(k)}
	\end{align*}
	where the last equality is by the fact that $\beta = \Theta(\alpha)$ and $c = \Theta(k/\alpha)$. 
\end{proof}

We are now ready to finalize the proof of Lemma~\ref{lem:cc-matching-detection}. 
\begin{proof}[Proof of Lemma~\ref{lem:cc-matching-detection}]
	Fix a set of blocks $\vecB$ and edges $\EB$ and assume that the event in Lemma~\ref{lem:matching-counting} happens. Notice that the mapping $\sigma_{\vecB}$ maps the set $\FC(\EB,\vecB)$ 
	to $2^{s}$ different choices $\Sigma_1,\ldots,\Sigma_{2^s}$. We define \event as the event that $\sigma_{\vecB}(\MA)$ maps to a class $\Sigma$ with $\card{\Sigma \cap \FC} \geq \card{\FC}/2^{2s}$. 
	The following claim on the probability of \event can be proven using a simple counting argument; the proof is deferred to after the current proof. 

	\begin{claim}\label{clm:event-whp}
	$\Pr\paren{\event} = 1-O(1/n)$. 
	\end{claim}
	Now fix a set $\Sigma$ that corresponds to the message Alice sent to Bob and suppose $\event$ happens. 
	 As argued earlier, given a message Bob can only output an edge in $\MA$ if it belongs to all matchings that are mapped to $\Sigma$, i.e., to $M_\Sigma$. By Lemma~\ref{lem:matching-counting}, the matching 
	$M_{\Sigma}$ contains at most $\ell = 2s + \Theta(k) = \Theta(s)$ edges (since $s = \Omega(k)$). 
	
	Now recall that the block $(P,Q)$ of Bob is chosen \emph{uniformly at random} from the blocks in $\vecB$, \emph{even} conditioned on a specific input matching $\MA$ to Alice; 
	this implies that in expectation $\ell/c = O(\alpha/k) \cdot \ell = (\alpha/k) \cdot O(s)$ edges of $M_{\Sigma}$ are 
	between $(P,Q)$. Consequently, Bob can only output $(\alpha/k) \cdot O(s)$ edges between $P$ and $Q$ in expectation. To finalize the proof, note that $\event$ and the event in Lemma~\ref{lem:matching-counting} happens
	w.h.p., and hence conditioning on these events can only change the expectation by an $O(1)$ additive factor. 
\end{proof}

For completeness, we provide a proof of Claim~\ref{clm:event-whp} here. 

\begin{proof}[Proof of Claim~\ref{clm:event-whp}]
	We say that $\Sigma_i$ in the partition $\Sigma_1,\ldots,\Sigma_{2^s}$ is \emph{light} iff $\card{\Sigma_i \cap \FC} < \card{\FC}/2^{2s}$. Since the matching $M_A$ is chosen uniformly at random, the probability that 
	$\sigma_{\vecB}(M_A)$ maps to some $\Sigma_i$ is exactly $\card{\Sigma_i \cap \FC}/\card{\FC}$. Hence, the probability that $M_A$ maps to some light set is
	at most 
	\begin{align*}
	2^{s} \cdot \frac{\card{\Sigma_i \cap \FC}}{\card{\FC}} \leq 2^{s} \cdot \frac{\card{\FC}}{\paren{2^{2s} \cdot \card{\FC}}} = \frac{1}{2^s} = O(1/n)
	\end{align*}
	where we used the fact that $s = \Omega(k) = \Omega(\alpha) = \Omega(\log{n})$. 
\end{proof}

\subsection{An $\Omega(nk/\alpha)$ Lower Bound on Communication Complexity of Vertex Cover}\label{sec:vc-dist}

We prove the following theorem on the simultaneous communication complexity of vertex cover, formalizing Result~\ref{res:lb-dist} for vertex cover. 

\begin{theorem}\label{thm:vc-lb}
	For any $\alpha$ between $\Omega(\log{n})$ and $o(\min\set{{n/k},k})$, the simultaneous communication complexity of $\alpha$-approximating the vertex cover problem in the random partition
	model with success probability at least $0.9$ is $\Omega(nk/\alpha)$. 
\end{theorem}

For simplicity of exposition, we prove the lower bound for protocols that can obtain a $c \cdot \alpha$ approximation for some small constant $c > 0$ to be determined later. 
By re-parametrizing $\alpha$ by a constant factor in the following, one can obtain the result for $\alpha$-approximation protocols as well. 
We again use the distribution $\distVC$ in Section~\ref{sec:lb-vc-coreset}.
For the convenience of the reader, we repeat the description of this distribution here. 

\textbox{Distribution $\distVC$. \textnormal{A hard input distribution for the vertex cover problem.}}{
\begin{itemize}
	\item Construct $G(L,R,E)$ (with $\card{L} = \card{R} = n$) as follows: 
		\begin{enumerate}
			\item Pick $A \subseteq L$ of size $n/\alpha$ uniformly at random. 
			\item Let $\EA$ be a set of edges chosen by picking each edge in $A \times R$ w.p. $k/2n$. 
			\item Pick a single vertex $\vstar$ uniformly at random from $\bA$ and let $\estar$ be an edge incident on $\vstar$ chosen uniformly at random. 
			\item Let $E:= \EA \cup \set{\estar}$. 
		\end{enumerate}
	\item \label{line:r-partitioning} Let $\Ei{1},\ldots,\Ei{k}$ be a \emph{random $k$-partitioning} of $E$ and let the input to player $\Ps{i}$ be the graph $\Gii(L,R,\Eii)$. 
\end{itemize}
} 

The intuition behind the proof is as follows. The distribution ensures that w.h.p., the input to each player $\Ps{i}$ contains $\Theta(n/\alpha)$ vertices in $L$ with degree exactly one. 
Let us denote this set with $D_i$. Now consider the input of the player $\Ps{\istar}$ which is given the edge $\estar$ also. It is easy to see that in this case, player $\Ps{\istar}$ is oblivious to which vertex 
in $D_i$ is $\ustar$; more formally, conditioned on the input $D_{\istar}$, the vertex $\ustar$ is chosen uniformly at random from $D_{\istar}$. This means the if $\Ps{\istar}$ communicates $o(\card{D_i}) = o(n/\alpha)$ bits, 
he is essentially not  ``revealing any information'' about $\vstar$ (or the other end point of $\estar$). On the other hand, as only $\Ps{\istar}$ has a knowledge about $\vstar$, this intuitively means that coordinator is not provided with 
enough information about $\vstar$ as well. This forces the coordinator to cover $\Omega(n)$ vertices to ensure that $\estar$ is also being covered. 

Making this intuition formal is complicated by the fact that the message by other players is still revealing ``some information'' about the input of player $\Ps{\istar}$, for instance, the identity of the set $A$ or the set 
of edges that may possibly be in the input of $\Ps{\istar}$. To overcome this, we show that by proper conditioning on some part of the input, one can embed an instance of the well-known set disjointness problem in this distribution. On a high level,
solving the disjointness on this embedded instance amounts to finding the vertex $\ustar$. This further allows us to design a reduction from our problem to the disjointness problem and prove the lower bound. Interestingly, while we 
consider the vertex cover only in the simultaneous model, our reduction requires a two-way communication between the players (however note that disjointness is still hard even in the two-way communication model). We now 
continue with the formal proof.


We can interpret the last line of distribution $\distVC$ as follows: Pick a random $k$-partitioning $\Epi{1},\ldots,\Epi{k}$ of \emph{all possible edges between $L$ and $R$}, and let $\Ei{i} = \Epi{i} \cap E$ for all $i \in [k]$. 
In the following, we assume that this \emph{initial partitioning} $\Epi{1},\ldots,\Epi{k}$ is public knowledge, as this allows us to reduce the dependence between the inputs of players which is crucial in our lower bound proof. Clearly, this 
assumption can only strengthen our results. Throughout the proof, we fix an arbitrary small constant $\eps > 0$. We say that the initial partitioning $\Epi{1},\ldots,\Epi{k}$ is \emph{$\eps$-balanced} if the degree of each vertex $v \in L$ in each graph $G(L,R,\Epi{i})$ for all $i \in [k]$ is in $(1 \pm \eps) \cdot n/k$ (note that this graph is \emph{not} the input graph to player $\Ps{i}$). As $n/k \geq \alpha = \Omega(\log{n})$, by Chernoff bound, 
any initial partitioning is $\eps$-balanced w.p. $1-O(1/n)$. Consequently, conditioning on this event is essentially not changing the 
distribution and hence for simplicity from now on, we always assume the inputs chosen from $\distVC$ satisfy the $\eps$-balanced property. 

We say that the player $\PS{i}$ (for $i \in [k]$) is the \emph{critical} player iff the edge $\estar$ is assigned to $\Eii$, i.e., it appears in the input to $\Ps{i}$. We use $\istar$ to denote the index of the critical player. 
In the following, we show that the critical player and the coordinator need to (implicitly) solve a ``hard'' communication task (named the \emph{Hidden Vertex Problem}, $\VertexSeeking$ for short) which requires a large communication from $\Ps{i}$. 

Fix $\ProtVC$ as a $\delta$-error $(c \cdot \alpha)$-approximation protocol for vertex cover over distribution $\distVC$ (for some sufficiently small constant $c > 0$ to be determined later). For any $i \in [k]$, 
let $\delta_i$ be the probability that $\ProtVC$ errs conditioned on $\istar = i$. A simple averaging argument ensures that for any $i \in [k]$ there exists a set $\Epi{i}$ such that 
$\Pr\paren{\ProtVC~\errs \mid \Epi{i}, \istar = i} \leq \delta_i$. We refer to such $\Epi{i}$ as a \emph{good} initial partition for $\Ps{i}$.

Fix a player $i \in [k]$ and a good initial partition $\Epi{i}$ for $i$. Let $\Dzi,\Doi$ and $\Dti$ be the set of vertices in $A$ that have degree, respectively, zero, one, and \emph{at least} two in the graph $\Gi{i}$. 
We further define $\Dmi := \Dzi \cup \Doi = A \setminus \Dti$. 

\begin{claim}\label{clm:exists-dti}
	For any $i \in [k]$ and any good initial partition $\Epi{i}$ for $i$, 
	there exists a set $\Dti$ with $\card{A \setminus \Dti} = \Omega(n/\alpha)$ such that, $\Pr\paren{\ProtVC~\errs \mid \Epi{i},\Dti,\istar=i} \leq \delta_i + o(1)$. 
\end{claim}
\begin{proof}
	Each vertex in $v \in A$ has degree more than $1$ in $\Gii$ independently and with some constant probability $p$ bounded away from $1$ (see the exact calculation in Claim~\ref{clm:obs-prob-twice}).
	Hence, in expectation $p \cdot n/\alpha$ vertices in $A$ have degree more than $1$. As $n/\alpha = \Omega(\log{n})$, by Chernoff bound plus a union bound, w.p. $1-o(1)$ at most $p \cdot n/\alpha + o(n/\alpha)$
	vertices in $A$ have degree more than $1$ in $\Gii$. The claim now follows immediately from this. 
\end{proof}

In the following, we further condition on a set $\Dti$ as in Claim~\ref{clm:exists-dti}. This implies that $A \setminus \Dti$, i.e., the set $\Dmi$ is a set of size $\Omega(n/\alpha)$ chosen uniformly at random 
from $L \setminus \Dti$. We are now ready to define the hidden vertex problem  in the one-way two-player communication model. 

\begin{problem}[The Hidden Vertex Problem (\VertexSeeking)]
	There are two disjoint sets $U$ and $V$ and a mapping $\sigma : U \rightarrow V$ known to both Alice and Bob. Bob is given a set $T \subseteq U$.
	Alice is given a set $S \subseteq T$ and a single vertex $\ustar$ chosen from $U \setminus T$ (identity of $\ustar$ is unknown to players). 
	Alice sends a single message to Bob and Bob needs to output two sets $X \subseteq U$ and $Y \subseteq V$ such that either $\ustar \in X$ or $\sigma(\ustar) \in Y$. The goal of the players 
	is to \emph{minimize the size} of $X \cup Y$. 
\end{problem}

Consider the following distribution $\distVS$ for \VertexSeeking: we sample an instance of vertex cover from $\distVC \mid \Epi{i},\Dti,\istar=i$, where $\Epi{i}$ is a good initial partition and $\Dti$ is a 
fixed set defined in Claim~\ref{clm:exists-dti}. We set $U = L \setminus \Dti$, $V = R$, and choose $\sigma : U \rightarrow V$ by mapping each $u \in U$ to one of the neighbors of $u$ (in $L$) in $\Epi{i}$ uniformly at random. 
Next, we set $T = \Dmi$ and $S = \Doi \cup \set{\vstar}$; this way the vertex $\ustar$ in \VertexSeeking is the special vertex $\vstar$ in distribution $\distVC$. We make the following simple observation about distribution 
$\distVS$. 

\begin{claim}\label{clm:obs-prob-twice}
	Each vertex in $T$ independently belongs to $S$ w.p. $(1 \pm O(\eps)) \cdot 1/3$. 
\end{claim}
\begin{proof}
	Note that $T = \Dmi$, i.e., the vertices that have degree $0$ or $1$ in $\Gii$, and vertices in $S \cap T$ are vertices that have degree exactly $1$ in $\Gii$. 
	Fix a vertex $v \in A$ and consider only conditioning on $\Epi{i}$. We know that since $\Epi{i}$ is a good initial partition, $v$ is incident on $(1 \pm \eps) \cdot n/k$ edges in $\Epi{i}$, and each of these edges appear
	independently in $\Gii$ w.p. $k/2n$ (by definition of $\distVC$). Let $d(v)$ denote the degree of $v$ in $\Gii$.  
	\begin{align*}
		\Pr\paren{d(v) = 0 \mid \Epi{i}} &= (1-\frac{k}{2n})^{\paren{1 \pm \eps} \cdot \frac{n}{k}} = e^{-\frac{1}{2}} \cdot (1\pm O(\eps)) \\
		\Pr\paren{d(v) = 1 \mid \Epi{i}} &= \paren{1 \pm \eps} \frac{n}{k} \cdot \paren{\frac{k}{2n}} \cdot (1-\frac{k}{2n})^{\paren{1 \pm \eps} \cdot \frac{n}{k}-1} \\
		&= \frac{1}{2} \cdot e^{-\frac{1}{2}} \cdot (1\pm O(\eps))
	\end{align*} 
	We can now conclude that $\Pr\paren{d(v) = 0 \mid \Epi{i}, \Dti} = (1 \pm O(\eps)) \cdot  2\cdot \paren{d(v) = 1 \mid \Epi{i}, \Dti}$, as conditioning on $\Dti$ imply that for each vertex $v \in T$, 
	$d(v) \in \set{0,1}$. The assertion of the claim now immediately follows. 
\end{proof}

We establish the following lower bound on the communication complexity of $\VertexSeeking$ on $\distVS$. 

\begin{lemma}[Communication Complexity of \VertexSeeking]\label{lem:cc-vertex-seeking}
	There exists a \emph{universal constant} $\cVS > 0$ such that any protocol for \VertexSeeking on the distribution $\distVS$ that computes an answer with $\card{X \cup Y} \leq \paren{\cVS \cdot n}$ w.p. at least $2/3$ needs 
	$\Omega(n/\alpha)$ communication. 
\end{lemma}

\newcommand{\criticalC}{\ensuremath{\mathcal{C}}}
We prove Lemma~\ref{lem:cc-vertex-seeking} in Section~\ref{sec:cc-vertex-seeking}. Before that, we establish Theorem~\ref{thm:vc-lb} using this lemma. 

\begin{proof}[Proof of Theorem~\ref{thm:vc-lb}]
	Recall that protocol $\ProtVC$ is a $\delta$-error $\paren{c \cdot \alpha}$-approximation protocol for vertex cover on $\distVC$. Note that w.p. $1-o(1)$, $\Epi{1},\ldots,\Epi{k}$ is an $\eps$-balanced 
	initial partitioning. Let $\event$ denote this event. Conditioned on $\event$, each edge in $G$ and in particular the edge $\estar$ belong to each player $i \in [k]$ w.p. $(1 \pm \eps) 1/k$. Hence, each player is the
	critical player w.p. $(1 \pm \eps) 1/k$. Let $\criticalC$ be the set of players such that $\delta_i \leq 2\delta$. We have $\card{\criticalC} \geq k/3$ as otherwise, 
	\begin{align*}
		\Pr\paren{\ProtVC~\errs \mid \event} &\geq \Pr\paren{\ProtVC~\errs \mid \event , \istar \notin \criticalC}
		 \cdot \Pr\paren{\istar \notin \criticalC} \\
		 &> 2\delta \cdot \card{\overline{\criticalC}} \cdot (1-\eps) \cdot 1/k > 4/3 \delta (1-\eps) > \delta
	\end{align*}
	for small enough $\eps > 0$, which contradicts the fact that $\Pr\paren{\ProtVC~\errs \mid \event} \leq \delta + o(1)$. 
		
	Fix an $i \in \criticalC$ and let $\Prot_i$ be the message sent by $\Ps{i}$ in protocol $\ProtVC$. 
	We use $\Prot_i$ to design a protocol $\Prot'$ for $\VertexSeeking$ on $\distVS$ (recall that $\distVS$ is a function of $\ProtVC$ and also index $i$). 
	Given an instance of $\VertexSeeking$ from $\distVS$, Alice and Bob create an instance of vertex cover sampled from $\distVC \mid \Epi{i},\Dti,\istar=i$ as follows: 
	\begin{enumerate}
		\item Alice plays the role of $\Ps{i}$ and Bob plays the role of all other players plus the coordinator. 
		\item Alice constructs the input of $\PS{i}$ (i.e., the graph $\Gii$) as follows: $(i)$ for each $u \in \Dti$, Alice samples the neighbors of $u$ from $\Epi{i}$ according to distribution $\distVC$, and $(ii)$ for 
		each $u \in \Doi \cup \set{\ustar}$, she adds the edge $(u,\sigma(u))$ to $\Gii$. 
		\item Bob constructs the inputs of all other players by letting the set $A = \Dti \cup T$ and sampling their inputs according to distribution $\distVC$. This is indeed possible since the input to players in $\distVC$ are 
		\emph{independent} conditioned on $A,\Epi{1},\ldots,\Epi{k},\istar$.   
	\end{enumerate}
	Next, Alice sends the message of $\Ps{i}$ to Bob and Bob simulates the messages of all other players (without any communication) and outputs the vertex cover computed by $\ProtVC$ 
	as the answer to the \VertexSeeking instance. 
	Using the definition of the distribution $\distVS$ one can verify
	 that the distribution of the instances sampled in this reduction matches distribution $\distVC \mid \Epi{i},\Dti,\istar=i$. Hence, since the minimum vertex cover size in $G$ is 
	at most $n/\alpha + 1$ (by picking $A \cup \set{\ustar}$), the output of $\Prot$ (i.e., the sets $X \cup Y$) is of size at most $c \cdot n$ w.p. $1-2\delta$ (as $i \in \criticalC$). 
	Moreover. since the edge $\estar$ in the vertex cover instance 
	corresponds to the pair $(\ustar,\sigma(\ustar))$ in the $\VertexSeeking$ instance, the returned solution is feasible. As $\delta \leq 0.1$, by picking the constant $c$ (in the $(c \cdot \alpha)$-approximation factor) 
	to be smaller than $\cVS$ (in Lemma~\ref{lem:cc-vertex-seeking}), we obtain that the size of ${\Prot_i}$ must be $\Omega(n/\alpha)$. Finally, since $\card{\criticalC} \geq k/3$ (i.e., there are at least $k/3$ choices for 
	player $\Ps{i}$), we obtain that the communication cost of $\ProtVC$ is $\Omega(nk/\alpha)$, proving the theorem. 
\end{proof}

We finish this section by noting that the bound stated in Theorem~\ref{thm:vc-lb} is in fact tight (up to an $O(\log{n})$ factor) for any approximation ratio $\alpha$. 

\begin{remark}\label{rem:vc-tight}
	The protocol in which the players group the vertices in the original graph into groups of size $\Theta(\alpha/\log{n})$ (deterministically but consistently across players) and then run
	 the algorithm in Theorem~\ref{thm:vc} on the resulting graph is an $\alpha$-approximation protocol with  $\Ot(nk/\alpha)$ communication for the minimum vertex cover problem. 
\end{remark}

Note that in Remark~\ref{rem:vc-tight}, we used the fact that Theorem~\ref{thm:vc} works even when the input graph has parallel edges, i.e., is a \emph{multi-graph}.

\subsubsection{Communication Complexity of \VertexSeeking}\label{sec:cc-vertex-seeking}
In this section, we prove Lemma~\ref{lem:cc-vertex-seeking} by a reduction from the well-known set disjointness in the two-player communication model. In this problem, 
Alice is given a set $A \subseteq [N]$ and Bob is given a set $B \subseteq [N]$ with the promise that $\card{A \cap B} \in \set{0,1}$ and their goal is to distinguish between these
two cases \emph{via two-way communication}. Let $\distDisj$ be the following distribution: start with $A = B = [N]$ and for each element $e \in A$, w.p. $1/2$, drop $e$ from both $A$ and $B$, w.p. $1/4$ drop 
$e$ from $A$, and with the remaining $1/4$ probability, drop $e$ from $B$. Next, pick an element $\estar \in [N]$ uniformly at random and w.p. $1/2$ add $\estar$ to both $A$ and $B$.  
It is known that solving disjointness under $\distDisj$ requires $\Omega(N)$ communication (see, e.g.,~\cite{Bar-YossefJKS02-S,Razborov92}). It also immediately follows from~\cite{Bar-YossefJKS02-S}
that if instead of dropping each element w.p. exactly $1/4$, we drop them w.p. $(1 \pm \eps) \cdot 1/4$ (for sufficiently small constant $\eps > 0$), the distribution still remains hard. 

Now let $\ProtVS$ be a $\delta$-error protocol for $\VertexSeeking$ on distribution $\distVS$. We use $\ProtVS$ to create a protocol $\Prot'$ for disjointness on distribution $\distDisj$. Note that while $\ProtVS$ is a 
one-way protocol, the protocol $\Prot'$ is allowed to use two-way communication. Given a pair of sets $(A,B)$ in $\distDisj$, we create an instance of $\VertexSeeking$ as follows: 
\begin{enumerate}
	\item Bob first communicates the \emph{size} of $B$ to Alice. 
	\item The players choose a set $Z$ of size $N = \card{T} + \card{B}$ vertices from $U$ uniformly at random and consider a fixed mapping between $[N]$ and $Z$; 
	note that $\card{T}$ is fixed in distribution $\distVS$ and $\card{B}$ is known at this point by both players.   
	\item Alice lets $S = A$ and Bob lets $T = Z \setminus B$ and they pick $\sigma$ uniformly at random from $\distVS$. 
	\item The players run $\ProtVS$; Bob computes the sets $X$ and $Y$ and let $B' = \paren{X \cup \sigma^{-1}(Y)} \cap B$. 
	\item If $\card{B'} > 3\cVS \cdot N$, Bob terminates the protocol. Otherwise the players run a lopsided set disjointness protocol (see, e.g.,~\cite{Patrascu11,DasguptaKS12}) for solving
	the disjointness instance $(A,B')$ (with error guarantee, say, $1/10$) and output the same answer as this protocol. 
\end{enumerate}

Whenever $(A,B)$ is a no instance of $\distDisj$, i.e., $\card{A \cap B} = 1$, the distribution of the instances constructed by $\Prot'$ is $\distVS$ (with $\ustar = A \cap B$). 
To see this, notice that for a fixed set $T$ in $\distVS$, each element in $T$ is in $S$ w.p. $1/3 \cdot (1 \pm O(\eps))$ and is outside $S$ with remaining probability (by Claim~\ref{clm:obs-prob-twice}).
 This is exactly the distribution of the set $[N] \setminus B$ in $\distDisj$ conditioned on $B$. The rest follows since we are choosing the set $T$ (by the random choice of $Z$) and $\sigma$ according to distribution 
 $\distVS$. 

On the other hand, when $\card{A \cap B} = 0$, the distribution of instances do \emph{not} correspond to $\distVS$. In fact, this is not 
even a valid instance of $\VertexSeeking$ as there is no element $\ustar$ in this instance.  This means that in this case, $\ProtVS$ may terminate, output a non-valid answer, or still
output two sets $X \subseteq U$ and $Y \subseteq V$ with $\card{X \cup Y} \leq \cVS \cdot n$. 
Unless the later happens, Bob is always able to distinguish this case as a Yes case of disjointness and solve the problem correctly (w.p. $1-\delta$). Hence, in the following, we assume
the worst case that $\ProtVS$ outputs two sets $X$ and $Y$ even if the instance created is not a legal input of $\VertexCollection$. We can now argue the following key lemma. 

\begin{lemma}\label{lem:intersection-size}
	In any instance $(S,T)$ of $\VertexCollection$ created by $\Prot'$, $\card{B'} \leq 3\cVS \cdot \card{B}$ w.h.p. 
\end{lemma}
\begin{proof}
	Consider the set $B^{-} := B \setminus \set{\ustar}$: this set is chosen from $U \setminus S \cup T$ uniformly at random. On the other hand, conditioned on 
	$S,T$ (and $\sigma$), i.e., all the inputs in distribution $\distVS$, the output of $\ProtVS$ are two fixed sets $X$ and $Y$ chosen independent of $B^{-}$. 
	This means that each vertex in $B^{-}$ belongs to $X$ w.p. $\card{X}/\card{U \setminus S \cup T}$. Similarly, each vertex in $\sigma(B^{-})$ also belongs to $Y$ w.p. $\card{Y}/\card{U \setminus S \cup T}$ (as $\sigma$ is a 
	random mapping).  This ensures that $\card{B^{-} \cap \paren{X \cup \sigma^{-1}(Y)}} \leq \paren{\card{X} + \card{Y}}/(n/2) \leq 2\cVS$ in expectation. A simple application of Chernoff bound finalizes the proof. 
\end{proof}

\begin{proof}[Proof of Lemma~\ref{lem:cc-vertex-seeking}]

	We first argue the correctness of the protocol $\Prot'$ and then bound its communication cost. 
	Clearly we have $B' \subseteq B$ and moreover in the no instances of disjointness, 
	the reduction ensures that $A \cap B \subseteq B'$; the reason is that $\ustar = A \cap B$ and since $\ProtVS$ is computing two sets $X$ and $Y$ which contain either $\ustar$ or $\sigma(\ustar)$, we obtain that $\ustar \in B'$.
	Consequently, the probability that $\Prot'$ errs is at most $1/3$ (if $\ProtVS$ errs), plus
	 $o(1)$ (if Bob terminates the protocol (by Lemma~\ref{lem:intersection-size})), plus $1/10$ (by error guarantee of the lopsided disjointness instance). This means that $\Prot'$ is a $\delta'$-error protocol for 
	 disjointness with $\delta' < 1/2$ (bounded away from half). 
	 
	 We now bound the communication cost of $\Prot'$. In the following, let $c$ be a constant such that communication complexity of disjointness on $\distDisj$ is at least $c \cdot N$. 
	Since if the protocol is not terminated, $\card{B'} \leq 3\cVS \cdot \card{B}$, the lopsided disjointness problem $(A,B')$ can be solved with $3\cVS \cdot O(\card{B}) = \cVS \cdot O(N)$ 
	communication (using, e.g., the protocol of H{\aa}stad and Wigderson~\cite{HastadW07}). Now assume by contradiction that cost $\ProtVS $ is $o(n/\alpha) = o(N)$. 
	This means that the total communication cost of $\Prot'$ is $O(\log{N}) + o(N) + \cVS \cdot O(N)$. By taking
	$\cVS \ll c$ but still a constant to suppress the constant in the $O(N)$ term above, the total cost of $\Prot'$ can be made
	smaller than $c \cdot N$. This contradicts the fact that communication complexity of disjointness on $\distDisj$ is at least $c \cdot N$, finalizing the proof. 
\end{proof}

\subsection*{Acknowledgements} 
The first author would like to thank Sepideh Mahabadi and Ali Vakilian for many helpful discussions in the earlier stages of this work. The authors are grateful to the anonymous reviewers of SPAA 2017 for many insightful comments and suggestions.

\bibliographystyle{abbrv}
\bibliography{general}

\begin{thebibliography}{10}

\bibitem{AbbarAIMV13}
S.~Abbar, S.~Amer{-}Yahia, P.~Indyk, S.~Mahabadi, and K.~R. Varadarajan.
\newblock Diverse near neighbor problem.
\newblock In {\em Symposuim on Computational Geometry 2013, SoCG '13, Rio de
  Janeiro, Brazil, June 17-20, 2013}, pages 207--214, 2013.

\bibitem{AgarwalHV04}
P.~K. Agarwal, S.~Har{-}Peled, and K.~R. Varadarajan.
\newblock Approximating extent measures of points.
\newblock {\em J. {ACM}}, 51(4):606--635, 2004.

\bibitem{AG13}
K.~J. Ahn and S.~Guha.
\newblock Linear programming in the semi-streaming model with application to
  the maximum matching problem.
\newblock {\em Inf. Comput.}, 222:59--79, 2013.

\bibitem{AhnG15}
K.~J. Ahn and S.~Guha.
\newblock Access to data and number of iterations: Dual primal algorithms for
  maximum matching under resource constraints.
\newblock In {\em Proceedings of the 27th {ACM} on Symposium on Parallelism in
  Algorithms and Architectures, {SPAA} 2015, Portland, OR, USA, June 13-15,
  2015}, pages 202--211, 2015.

\bibitem{AhnGM12Linear}
K.~J. Ahn, S.~Guha, and A.~McGregor.
\newblock Analyzing graph structure via linear measurements.
\newblock In {\em Proceedings of the Twenty-third Annual ACM-SIAM Symposium on
  Discrete Algorithms}, SODA '12, pages 459--467. SIAM, 2012.

\bibitem{AGM12}
K.~J. Ahn, S.~Guha, and A.~McGregor.
\newblock Graph sketches: sparsification, spanners, and subgraphs.
\newblock In {\em Proceedings of the 31st {ACM} {SIGMOD-SIGACT-SIGART}
  Symposium on Principles of Database Systems, {PODS} 2012, Scottsdale, AZ,
  USA, May 20-24, 2012}, pages 5--14, 2012.

\bibitem{AiHLW16}
Y.~Ai, W.~Hu, Y.~Li, and D.~P. Woodruff.
\newblock New characterizations in turnstile streams with applications.
\newblock In {\em 31st Conference on Computational Complexity, {CCC} 2016, May
  29 to June 1, 2016, Tokyo, Japan}, pages 20:1--20:22, 2016.

\bibitem{AlonNRW15}
N.~Alon, N.~Nisan, R.~Raz, and O.~Weinstein.
\newblock Welfare maximization with limited interaction.
\newblock In {\em {IEEE} 56th Annual Symposium on Foundations of Computer
  Science, {FOCS} 2015, Berkeley, CA, USA, 17-20 October, 2015}, pages
  1499--1512, 2015.

\bibitem{AssadiKL17}
S.~Assadi, S.~Khanna, and Y.~Li.
\newblock On estimating maximum matching size in graph streams.
\newblock In {\em Proceedings of the Twenty-Eighth Annual {ACM-SIAM} Symposium
  on Discrete Algorithms, {SODA} 2017, Barcelona, Spain, Hotel Porta Fira,
  January 16-19}, pages 1723--1742, 2017.

\bibitem{AssadiKLY16}
S.~Assadi, S.~Khanna, Y.~Li, and G.~Yaroslavtsev.
\newblock Maximum matchings in dynamic graph streams and the simultaneous
  communication model.
\newblock In {\em Proceedings of the Twenty-Seventh Annual {ACM-SIAM} Symposium
  on Discrete Algorithms, {SODA} 2016, Arlington, VA, USA, January 10-12,
  2016}, pages 1345--1364, 2016.

\bibitem{BadanidiyuruMKK14}
A.~Badanidiyuru, B.~Mirzasoleiman, A.~Karbasi, and A.~Krause.
\newblock Streaming submodular maximization: massive data summarization on the
  fly.
\newblock In {\em The 20th {ACM} {SIGKDD} International Conference on Knowledge
  Discovery and Data Mining, {KDD} '14, New York, NY, {USA} - August 24 - 27,
  2014}, pages 671--680, 2014.

\bibitem{BalcanEL13}
M.~Balcan, S.~Ehrlich, and Y.~Liang.
\newblock Distributed k-means and k-median clustering on general communication
  topologies.
\newblock In {\em Advances in Neural Information Processing Systems 26: 27th
  Annual Conference on Neural Information Processing Systems 2013. Proceedings
  of a meeting held December 5-8, 2013, Lake Tahoe, Nevada, United States.},
  pages 1995--2003, 2013.

\bibitem{Bar-YossefJKS02-S}
Z.~Bar{-}Yossef, T.~S. Jayram, R.~Kumar, and D.~Sivakumar.
\newblock An information statistics approach to data stream and communication
  complexity.
\newblock In {\em 43rd Symposium on Foundations of Computer Science {(FOCS}
  2002), 16-19 November 2002, Vancouver, BC, Canada, Proceedings}, pages
  209--218, 2002.

\bibitem{BaswanaGS15}
S.~Baswana, M.~Gupta, and S.~Sen.
\newblock Fully dynamic maximal matching in o(log n) update time.
\newblock {\em {SIAM} J. Comput.}, 44(1):88--113, 2015.

\bibitem{BateniBLM14}
M.~Bateni, A.~Bhaskara, S.~Lattanzi, and V.~S. Mirrokni.
\newblock Distributed balanced clustering via mapping coresets.
\newblock In {\em Advances in Neural Information Processing Systems 27: Annual
  Conference on Neural Information Processing Systems 2014, December 8-13 2014,
  Montreal, Quebec, Canada}, pages 2591--2599, 2014.

\bibitem{BhattacharyaHI15}
S.~Bhattacharya, M.~Henzinger, and G.~F. Italiano.
\newblock Deterministic fully dynamic data structures for vertex cover and
  matching.
\newblock In {\em Proceedings of the Twenty-Sixth Annual {ACM-SIAM} Symposium
  on Discrete Algorithms, {SODA} 2015, San Diego, CA, USA, January 4-6, 2015},
  pages 785--804, 2015.

\bibitem{BhattacharyaHNT15}
S.~Bhattacharya, M.~Henzinger, D.~Nanongkai, and C.~E. Tsourakakis.
\newblock Space- and time-efficient algorithm for maintaining dense subgraphs
  on one-pass dynamic streams.
\newblock In {\em Proceedings of the Forty-Seventh Annual {ACM} on Symposium on
  Theory of Computing, {STOC} 2015, Portland, OR, USA, June 14-17, 2015}, pages
  173--182, 2015.

\bibitem{BulteauFKP16}
L.~Bulteau, V.~Froese, K.~Kutzkov, and R.~Pagh.
\newblock Triangle counting in dynamic graph streams.
\newblock {\em Algorithmica}, 76(1):259--278, 2016.

\bibitem{ChakrabartiCM08}
A.~Chakrabarti, G.~Cormode, and A.~McGregor.
\newblock Robust lower bounds for communication and stream computation.
\newblock In {\em Proceedings of the 40th Annual {ACM} Symposium on Theory of
  Computing, Victoria, British Columbia, Canada, May 17-20, 2008}, pages
  641--650, 2008.

\bibitem{ChitnisCEHMMV16}
R.~Chitnis, G.~Cormode, H.~Esfandiari, M.~Hajiaghayi, A.~McGregor,
  M.~Monemizadeh, and S.~Vorotnikova.
\newblock Kernelization via sampling with applications to finding matchings and
  related problems in dynamic graph streams.
\newblock In {\em Proceedings of the Twenty-Seventh Annual {ACM-SIAM} Symposium
  on Discrete Algorithms, {SODA} 2016, Arlington, VA, USA, January 10-12,
  2016}, pages 1326--1344, 2016.

\bibitem{ChitnisCHM15}
R.~H. Chitnis, G.~Cormode, M.~T. Hajiaghayi, and M.~Monemizadeh.
\newblock Parameterized streaming: Maximal matching and vertex cover.
\newblock In {\em Proceedings of the Twenty-Sixth Annual {ACM-SIAM} Symposium
  on Discrete Algorithms, {SODA} 2015, San Diego, CA, USA, January 4-6, 2015},
  pages 1234--1251, 2015.

\bibitem{CS14}
M.~Crouch and D.~S. Stubbs.
\newblock Improved streaming algorithms for weighted matching, via unweighted
  matching.
\newblock In {\em Approximation, Randomization, and Combinatorial Optimization.
  Algorithms and Techniques, {APPROX/RANDOM} 2014, September 4-6, 2014,
  Barcelona, Spain}, pages 96--104, 2014.

\bibitem{DasguptaKS12}
A.~Dasgupta, R.~Kumar, and D.~Sivakumar.
\newblock Sparse and lopsided set disjointness via information theory.
\newblock In {\em Approximation, Randomization, and Combinatorial Optimization.
  Algorithms and Techniques - 15th International Workshop, {APPROX} 2012, and
  16th International Workshop, {RANDOM} 2012, Cambridge, MA, USA, August 15-17,
  2012. Proceedings}, pages 517--528, 2012.

\bibitem{DNO14}
S.~Dobzinski, N.~Nisan, and S.~Oren.
\newblock Economic efficiency requires interaction.
\newblock In {\em Symposium on Theory of Computing, {STOC} 2014, New York, NY,
  USA, May 31 - June 03, 2014}, pages 233--242, 2014.

\bibitem{ConcentrationBook}
D.~P. Dubhashi and A.~Panconesi.
\newblock {\em Concentration of Measure for the Analysis of Randomized
  Algorithms}.
\newblock Cambridge University Press, 2009.

\bibitem{EKS09}
S.~Eggert, L.~Kliemann, and A.~Srivastav.
\newblock Bipartite graph matchings in the semi-streaming model.
\newblock In {\em Algorithms - {ESA} 2009, 17th Annual European Symposium,
  Copenhagen, Denmark, September 7-9, 2009. Proceedings}, pages 492--503, 2009.

\bibitem{EpsteinLMS11}
L.~Epstein, A.~Levin, J.~Mestre, and D.~Segev.
\newblock Improved approximation guarantees for weighted matching in the
  semi-streaming model.
\newblock {\em {SIAM} J. Discrete Math.}, 25(3):1251--1265, 2011.

\bibitem{EsfandiariHM16}
H.~Esfandiari, M.~Hajiaghayi, and M.~Monemizadeh.
\newblock Finding large matchings in semi-streaming.
\newblock In {\em {IEEE} International Conference on Data Mining Workshops,
  {ICDM} Workshops 2016, December 12-15, 2016, Barcelona, Spain.}, pages
  608--614, 2016.

\bibitem{EsfandiariHLMO15}
H.~Esfandiari, M.~T. Hajiaghayi, V.~Liaghat, M.~Monemizadeh, and K.~Onak.
\newblock Streaming algorithms for estimating the matching size in planar
  graphs and beyond.
\newblock In {\em Proceedings of the Twenty-Sixth Annual {ACM-SIAM} Symposium
  on Discrete Algorithms, {SODA} 2015, San Diego, CA, USA, January 4-6, 2015},
  pages 1217--1233, 2015.

\bibitem{FKMSZ05}
J.~Feigenbaum, S.~Kannan, A.~McGregor, S.~Suri, and J.~Zhang.
\newblock On graph problems in a semi-streaming model.
\newblock {\em Theor. Comput. Sci.}, 348(2-3):207--216, 2005.

\bibitem{GoelKK12}
A.~Goel, M.~Kapralov, and S.~Khanna.
\newblock On the communication and streaming complexity of maximum bipartite
  matching.
\newblock In {\em Proceedings of the Twenty-third Annual ACM-SIAM Symposium on
  Discrete Algorithms}, SODA '12, pages 468--485. SIAM, 2012.

\bibitem{GO13}
V.~Guruswami and K.~Onak.
\newblock Superlinear lower bounds for multipass graph processing.
\newblock In {\em Proceedings of the 28th Conference on Computational
  Complexity, {CCC} 2013, K.lo Alto, California, USA, 5-7 June, 2013}, pages
  287--298, 2013.

\bibitem{HassidimKNO09}
A.~Hassidim, J.~A. Kelner, H.~N. Nguyen, and K.~Onak.
\newblock Local graph partitions for approximation and testing.
\newblock In {\em 50th Annual {IEEE} Symposium on Foundations of Computer
  Science, {FOCS} 2009, October 25-27, 2009, Atlanta, Georgia, {USA}}, pages
  22--31, 2009.

\bibitem{HastadW07}
J.~H{\aa}stad and A.~Wigderson.
\newblock The randomized communication complexity of set disjointness.
\newblock {\em Theory of Computing}, 3(1):211--219, 2007.

\bibitem{HuangRVZ15}
Z.~Huang, B.~Radunovic, M.~Vojnovic, and Q.~Zhang.
\newblock Communication complexity of approximate matching in distributed
  graphs.
\newblock In {\em 32nd International Symposium on Theoretical Aspects of
  Computer Science, {STACS} 2015, March 4-7, 2015, Garching, Germany}, pages
  460--473, 2015.

\bibitem{IndykMMM14}
P.~Indyk, S.~Mahabadi, M.~Mahdian, and V.~S. Mirrokni.
\newblock Composable core-sets for diversity and coverage maximization.
\newblock In {\em Proceedings of the 33rd {ACM} {SIGMOD-SIGACT-SIGART}
  Symposium on Principles of Database Systems, PODS'14, Snowbird, UT, USA, June
  22-27, 2014}, pages 100--108, 2014.

\bibitem{Kapralov13}
M.~Kapralov.
\newblock Better bounds for matchings in the streaming model.
\newblock In {\em Proceedings of the Twenty-Fourth Annual {ACM-SIAM} Symposium
  on Discrete Algorithms, {SODA} 2013, New Orleans, Louisiana, USA, January
  6-8, 2013}, pages 1679--1697, 2013.

\bibitem{KapralovKS14}
M.~Kapralov, S.~Khanna, and M.~Sudan.
\newblock Approximating matching size from random streams.
\newblock In {\em Proceedings of the Twenty-Fifth Annual {ACM-SIAM} Symposium
  on Discrete Algorithms, {SODA} 2014, Portland, Oregon, USA, January 5-7,
  2014}, pages 734--751, 2014.

\bibitem{KapralovKS15}
M.~Kapralov, S.~Khanna, and M.~Sudan.
\newblock Streaming lower bounds for approximating {MAX-CUT}.
\newblock In {\em SODA}, 2015.

\bibitem{KapralovLMMS14}
M.~Kapralov, Y.~T. Lee, C.~Musco, C.~Musco, and A.~Sidford.
\newblock Single pass spectral sparsification in dynamic streams.
\newblock In {\em 55th {IEEE} Annual Symposium on Foundations of Computer
  Science, {FOCS} 2014, Philadelphia, PA, USA, October 18-21, 2014}, pages
  561--570, 2014.

\bibitem{KW14}
M.~Kapralov and D.~Woodruff.
\newblock Spanners and sparsifiers in dynamic streams.
\newblock {\em PODC}, 2014.

\bibitem{KSV10}
H.~J. Karloff, S.~Suri, and S.~Vassilvitskii.
\newblock A model of computation for mapreduce.
\newblock In {\em Proceedings of the Twenty-First Annual {ACM-SIAM} Symposium
  on Discrete Algorithms, {SODA} 2010, Austin, Texas, USA, January 17-19,
  2010}, pages 938--948, 2010.

\bibitem{Konrad15}
C.~Konrad.
\newblock Maximum matching in turnstile streams.
\newblock In {\em Algorithms - {ESA} 2015 - 23rd Annual European Symposium,
  Patras, Greece, September 14-16, 2015, Proceedings}, pages 840--852, 2015.

\bibitem{KonradMM12}
C.~Konrad, F.~Magniez, and C.~Mathieu.
\newblock Maximum matching in semi-streaming with few passes.
\newblock In {\em Approximation, Randomization, and Combinatorial Optimization.
  Algorithms and Techniques - 15th International Workshop, {APPROX} 2012, and
  16th International Workshop, {RANDOM} 2012, Cambridge, MA, USA, August 15-17,
  2012. Proceedings}, pages 231--242, 2012.

\bibitem{KN97}
E.~Kushilevitz and N.~Nisan.
\newblock {\em Communication complexity}.
\newblock Cambridge University Press, 1997.

\bibitem{LMSV11}
S.~Lattanzi, B.~Moseley, S.~Suri, and S.~Vassilvitskii.
\newblock Filtering: a method for solving graph problems in mapreduce.
\newblock In {\em {SPAA} 2011: Proceedings of the 23rd Annual {ACM} Symposium
  on Parallelism in Algorithms and Architectures, San Jose, CA, USA, June 4-6,
  2011 (Co-located with {FCRC} 2011)}, pages 85--94, 2011.

\bibitem{LiNW14}
Y.~Li, H.~L. Nguyen, and D.~P. Woodruff.
\newblock Turnstile streaming algorithms might as well be linear sketches.
\newblock In {\em Symposium on Theory of Computing, {STOC} 2014, New York, NY,
  USA, May 31 - June 03, 2014}, pages 174--183, 2014.

\bibitem{M05}
A.~McGregor.
\newblock Finding graph matchings in data streams.
\newblock In {\em Approximation, Randomization and Combinatorial Optimization,
  Algorithms and Techniques, 8th International Workshop on Approximation
  Algorithms for Combinatorial Optimization Problems, {APPROX} 2005 and 9th
  InternationalWorkshop on Randomization and Computation, {RANDOM} 2005,
  Berkeley, CA, USA, August 22-24, 2005, Proceedings}, pages 170--181, 2005.

\bibitem{M14}
A.~McGregor.
\newblock Graph stream algorithms: a survey.
\newblock {\em {SIGMOD} Record}, 43(1):9--20, 2014.

\bibitem{McGregorTVV15}
A.~McGregor, D.~Tench, S.~Vorotnikova, and H.~T. Vu.
\newblock Densest subgraph in dynamic graph streams.
\newblock In {\em Mathematical Foundations of Computer Science 2015 - 40th
  International Symposium, {MFCS} 2015, Milan, Italy, August 24-28, 2015,
  Proceedings, Part {II}}, pages 472--482, 2015.

\bibitem{McGregorV16}
A.~McGregor and S.~Vorotnikova.
\newblock Planar matching in streams revisited.
\newblock {\em To appear in Approximation, Randomization, and Combinatorial
  Optimization. Algorithms and Techniques - 19th International Workshop,
  {APPROX} 2016, and 20th International Workshop, {RANDOM} 2016}, 2016.

\bibitem{MirrokniZ15}
V.~S. Mirrokni and M.~Zadimoghaddam.
\newblock Randomized composable core-sets for distributed submodular
  maximization.
\newblock In {\em Proceedings of the Forty-Seventh Annual {ACM} on Symposium on
  Theory of Computing, {STOC} 2015, Portland, OR, USA, June 14-17, 2015}, pages
  153--162, 2015.

\bibitem{MirzasoleimanKSK13}
B.~Mirzasoleiman, A.~Karbasi, R.~Sarkar, and A.~Krause.
\newblock Distributed submodular maximization: Identifying representative
  elements in massive data.
\newblock In {\em Advances in Neural Information Processing Systems 26: 27th
  Annual Conference on Neural Information Processing Systems 2013. Proceedings
  of a meeting held December 5-8, 2013, Lake Tahoe, Nevada, United States.},
  pages 2049--2057, 2013.

\bibitem{muthukrishnan2005data}
S.~Muthukrishnan.
\newblock {\em Data streams: Algorithms and applications}.
\newblock Now Publishers Inc, 2005.

\bibitem{NeimanS13}
O.~Neiman and S.~Solomon.
\newblock Simple deterministic algorithms for fully dynamic maximal matching.
\newblock In {\em Symposium on Theory of Computing Conference, STOC'13, Palo
  Alto, CA, USA, June 1-4, 2013}, pages 745--754, 2013.

\bibitem{NguyenO08}
H.~N. Nguyen and K.~Onak.
\newblock Constant-time approximation algorithms via local improvements.
\newblock In {\em 49th Annual {IEEE} Symposium on Foundations of Computer
  Science, {FOCS} 2008, October 25-28, 2008, Philadelphia, PA, {USA}}, pages
  327--336, 2008.

\bibitem{OnakRRR12}
K.~Onak, D.~Ron, M.~Rosen, and R.~Rubinfeld.
\newblock A near-optimal sublinear-time algorithm for approximating the minimum
  vertex cover size.
\newblock In {\em Proceedings of the Twenty-Third Annual {ACM-SIAM} Symposium
  on Discrete Algorithms, {SODA} 2012, Kyoto, Japan, January 17-19, 2012},
  pages 1123--1131, 2012.

\bibitem{OnakR10}
K.~Onak and R.~Rubinfeld.
\newblock Maintaining a large matching and a small vertex cover.
\newblock In {\em Proceedings of the 42nd {ACM} Symposium on Theory of
  Computing, {STOC} 2010, Cambridge, Massachusetts, USA, 5-8 June 2010}, pages
  457--464, 2010.

\bibitem{ParnasR07}
M.~Parnas and D.~Ron.
\newblock Approximating the minimum vertex cover in sublinear time and a
  connection to distributed algorithms.
\newblock {\em Theor. Comput. Sci.}, 381(1-3):183--196, 2007.

\bibitem{Patrascu11}
M.~Patrascu.
\newblock Unifying the landscape of cell-probe lower bounds.
\newblock {\em {SIAM} J. Comput.}, 40(3):827--847, 2011.

\bibitem{PazS17}
A.~Paz and G.~Schwartzman.
\newblock A (2 + {\eps})-approximation for maximum weight matching in the
  semi-streaming model.
\newblock In {\em Proceedings of the Twenty-Eighth Annual {ACM-SIAM} Symposium
  on Discrete Algorithms, {SODA} 2017, Barcelona, Spain, Hotel Porta Fira,
  January 16-19}, pages 2153--2161, 2017.

\bibitem{PhillipsVZ12}
J.~M. Phillips, E.~Verbin, and Q.~Zhang.
\newblock Lower bounds for number-in-hand multiparty communication complexity,
  made easy.
\newblock In {\em Proceedings of the Twenty-Third Annual {ACM-SIAM} Symposium
  on Discrete Algorithms, {SODA} 2012, Kyoto, Japan, January 17-19, 2012},
  pages 486--501, 2012.

\bibitem{Razborov92}
A.~A. Razborov.
\newblock On the distributional complexity of disjointness.
\newblock {\em Theor. Comput. Sci.}, 106(2):385--390, 1992.

\bibitem{Solomon16}
S.~Solomon.
\newblock Fully dynamic maximal matching in constant update time.
\newblock {\em CoRR}, abs/1604.08491. To appear in FOCS, 2016.

\bibitem{Yao87}
A.~C. Yao.
\newblock Lower bounds to randomized algorithms for graph properties (extended
  abstract).
\newblock In {\em 28th Annual Symposium on Foundations of Computer Science, Los
  Angeles, California, USA, 27-29 October 1987}, pages 393--400, 1987.

\bibitem{YoshidaYI12}
Y.~Yoshida, M.~Yamamoto, and H.~Ito.
\newblock Improved constant-time approximation algorithms for maximum matchings
  and other optimization problems.
\newblock {\em {SIAM} J. Comput.}, 41(4):1074--1093, 2012.

\end{thebibliography}
\clearpage
\appendix
\section{Proof of Lemma~\ref{lem:matching-dist-large}}\label{app:large-induced-matching}

Consider the edges of $\Eab$ assigned to the graph $\Gii$. It is easy to see that by the choice of $\Eab$ and further partitioning the edges between the $k$ players, 
the graph $\Gii$ on the set of edges $\Eab^{(i)}$ forms a random bipartite graph. Hence, proving Lemma~\ref{lem:matching-dist-large} reduces to proving the following property of 
random bipartite graphs.

Let ${\cal G}(n,n,1/n)$ denote the family of random bipartite graphs where each side of the bipartition contains $n$ vertices, and each edge is present w.p. $1/n$. 
We will show that if we sample a random graph $G \in {\cal G}(n,n,1/n)$, then w.p. at least $1 - 1/n^2$, it contains an {\em induced matching} of size $\Omega(n)$. We emphasize here that the notion of induced matching is with respect 
to the entire graph and not only with respect to the vertices included in the induced matching.

Our proof will use the following pair of elementary propositions.

\begin{proposition}
\label{prop:random1}
Suppose we assign $N$ balls uniformly at random to $M > N$ bins. Let $B$ be an arbitrary fixed subset of bins. Then with probability at least $1 - \frac{1}{N^3}$, there are at least $\left( \frac{|B|}{M} \right) \cdot \frac{N}{e} - o(N)$ bins in $B$ that contain exactly one ball, assuming $N$ is sufficiently large. 
\end{proposition}
\begin{proof}
We arbitrarily number the balls $1$ through $N$ and the bins $1$ through $M$. W.l.o.g. assume that the bins in $B$ are numbered $1$ through $|B|$.
For $1 \le i \le |B|$, let $Z_i$ be the $0/1$ random variable that indicates whether or not bin $i \in B$ receives exactly one ball, and furthermore, let $Z = \sum_{i=1}^{|B|} Z_i$. Then

$${\rm Pr}[Z_i = 1] = {N \choose 1}\cdot \left( \frac{1}{M} \right)  \cdot \left(1 - \frac{1}{M}\right)^{N-1} \ge \frac{N}{M}\cdot\left( \frac{1}{e} - o(1) \right).$$

Hence $E[Z] \ge \left( \frac{|B|}{M} \right) \cdot \frac{N}{e} - o(N)$. We now wish to argue that the value of $Z$ is concentrated around its expectation. However, we can not directly invoke the standard Chernoff bound since the variables $Z_i$'s are not independent. We will instead utilize the more general version stated in Proposition~\ref{prop:bounded-differences}.

Let $X_{j} \in [1..M]$ denote the index of the bin in which the $j_{th}$ ball lands. Given the variables $X_1, X_2, ..., X_N$, we can define the function $f(X_1, X_2, ...., X_N)$ to be the number of bins in $B$ that receive exactly one ball. Note that $f$ is completely determined by the variables $X_1, X_2, ..., X_N$ and that $E[f] = E[Z] \ge \left( \frac{|B|}{M} \right) \cdot \frac{N}{e} - o(N)$. It is easy to see that the function $f$ satisfies the Lipschitz property with $d=2$ since changing the assignment of any single ball, can reduce or increase the number of bins in $B$ with exactly one ball by at most $2$. We can thus invoke Proposition~\ref{prop:bounded-differences} with $t = 4\sqrt{N \ln N}$, completing the proof.
\end{proof}

\begin{proposition}
\label{prop:random2}
For sufficiently large $n$, with probability at least $1 - 1/n^3$, a graph $G(L \cup R, E)$ drawn from ${\cal G}(n,n,1/n)$ satisfies the following properties:
\begin{itemize}
\item[(a)] The set $S \subseteq L$ of all vertices in $L$ with degree exactly one in $G$ has size $n/e \pm o(n)$. 
\item[(b)] The set $T \subseteq R$ of vertices defined as all vertices in $R$ with no edges to $L \setminus S$ has size at least $n/e - o(n)$.
\end{itemize}
\end{proposition}
\begin{proof}
To see property (a), let us define $0/1$ random variables $X_1, X_2, ..., X_n$ where $X_i = 1$ iff vertex $i \in L$ has degree exactly one in $G$. Then ${\rm Pr}[X_i=1] = (1-\frac{1}{n})^{n-1} = 1/e - o(1)$ for sufficiently large $n$. Thus $E[\sum_{i=1}^{n} X_i] = n/e - o(n)$, and using Chernoff bound (Proposition~\ref{prop:chernoff} with $t = 4\sqrt{n \ln n}$) implies that with probability at least $1 - 2/n^4$, there is a set $S \subseteq L$ of size $n/e \pm o(n)$ whose vertices have degree exactly one in $G$.

To see property (b), fix a set $S$ of degree $1$ vertices in $L$. Let us define $0/1$ random variables $Y_1, Y_2, ..., Y_n$ where $Y_i = 1$ iff vertex $i \in R$ receives no edges from vertices in $L \setminus S$. 
Then ${\rm Pr}[Y_i=1] = (1-\frac{1}{n})^{|L \setminus S|} \ge 1/e - o(1)$ for sufficiently large $n$. Thus $E[\sum_{i=1}^{n} Y_i] \ge n/e - o(n)$, and using Chernoff bound (Proposition~\ref{prop:chernoff} with $t = 4\sqrt{n \ln n}$) implies that with probability at least $1 - 2/n^4$, there is a set $T \subseteq R$ of size at least $n/e - o(n)$ whose vertices 
do not have any edges to $L \setminus S$.

Thus both properties (a) and (b) hold with probability at least $1 - 1/n^3$, as desired.
\end{proof}

\begin{lemma}\label{lem:random-graph}
Let $G(L \cup R, E)$ be drawn from ${\cal G}(n,n,1/n)$. Then for sufficiently large $n$, with probability at least $1 - 1/n^2$, $G$ contains an {\em induced matching} of size $n/e^3 -o(n)$.
\end{lemma}
\begin{proof}
By Proposition~\ref{prop:random2}, we know with probability at least $1 - 1/n^3$, the graph $G(L \cup R, E)$ satisfies properties (a) and (b). We will assume from here on that this event, denoted by ${\cal E}$, has occurred. We first 
observe that conditioned on the event ${\cal E}$, and for any choice of sets $S$ and $T$ as defined in Proposition~\ref{prop:random2} as well as edges from the set $L \setminus S$ to $R \setminus T$, sampling a graph $G$ from $
{\cal G}(n,n,1/n)$ is equivalent to assigning each vertex in $S$ a uniformly at random neighbor in $R$. 

Now invoking Proposition~\ref{prop:random1}, with $N = |S|$, $B = T$, we know that with 
probability at least $1 - O(1/n^3)$, there is a set $T' \subseteq T$ of size at least  

$$\frac{|T|}{n} \cdot \frac{|S|}{e} - o(|S|) = \left( \frac{1}{e} - o(1) \right) \cdot \left( \frac{n/e - o(n)}{e} \right) - o(n) \ge \frac{n}{e^3} - o(n)$$

such that each vertex in $T'$ receives exactly one ball from $S$ (i.e. receives exactly one edge from the vertices in $S$).
Let $S' \subseteq S$ be the set of vertices that ``supply a ball'' (i.e. an edge) to vertices in $T'$. Since by definition the vertices in $T$ receive edges only from $S$, and since all vertices in $S$ have degree exactly one, the set $S' \cup 
T'$ of vertices induces a matching of size at least $n/e^3 - o(n)$ in $G$, as asserted in the lemma.
\end{proof}

The lower bound in Lemma~\ref{lem:matching-dist-large} now follows from Lemma~\ref{lem:random-graph} for the family of bipartite graphs with $n/\alpha$ vertices on each side. 
The upper bound is a simple application of Chernoff bound on the number of edges from $\Ebab$ that are assigned to $\Gii$.

\end{document}